%% file: main.tex
\documentclass[sigconf, table]{acmart}

\input{includes}
\input{defines}

\RequirePackage{OpenProc-Mods}

\AddToHook{cmd/maketitle/before}{%
  \input{edbt-macros-v30-n1}}

\def\BibTeX{{\rm B\kern-.05em{\sc i\kern-.025em b}\kern-.08em
    T\kern-.1667em\lower.7ex\hbox{E}\kern-.125emX}}

\begin{document}


\OPtitle{$\text{C}^2$: Cache-Conscious Succinct Tries with \\Adaptive Unary Path Compression}



\OPauthor{Kepan Zhang}
\affiliation{%
  \institution{Georgia Institute of Technology}
  \country{USA}
}
\email{zhangkepan03@gmail.com}

\OPauthor{Tiancheng Zhao}
\affiliation{%
  \institution{Georgia Institute of Technology}
    \country{USA}
}
\email{tzhao350@gatech.edu}

\OPauthor[0000-0003-2232-3305]{Helen Xu}
\affiliation{%
   \institution{Georgia Institute of Technology}
     \country{USA}
 }
 \email{hxu615@gatech.edu}


\renewcommand{\shortauthors}{Zhang et al.}



\input{new-abstract}

 \keywords{succinct, tries, cache, optimization, path, compression}

\maketitle

\input{introduction}
\input{prelim}
\input{bv-design}
\input{unary-path}
\input{experiments}
\input{conclusion}

\clearpage
\bibliographystyle{ACM-Reference-Format}
\bibliography{sample}


\end{document}

%% file: includes.tex
\usepackage{amsmath,amsfonts}
\usepackage{algorithm,algpseudocode}
\usepackage{graphicx}
\usepackage{textcomp}
\usepackage{xspace}
\usepackage{url}
\usepackage{tabularx}
\usepackage{multirow}
\usepackage{forest}
\usepackage{tikz}
\usepackage{xcolor}
\usepackage{pgfplots}
\usepackage{booktabs}
\usepackage{algpseudocode}
\usepackage{balance}
\usepackage{microtype}
\usepackage{enumitem}
\usepackage{subcaption}
\usepackage{makecell}
\usetikzlibrary{pgfplots.groupplots}
\usepackage{pgfplotstable}

\pgfplotsset{compat=1.9}
\usetikzlibrary{shapes.geometric,positioning,arrows.meta, shadows, shadows.blur}
\tikzset{every shadow/.style={shadow xshift=5pt,shadow yshift=-5pt}}
\def\BibTeX{{\rm B\kern-.05em{\sc i\kern-.025em b}\kern-.08em
    T\kern-.1667em\lower.7ex\hbox{E}\kern-.125emX}}


%% file: defines.tex
\newcommand{\para}[1]{\smallskip\noindent\textbf{#1.}}
\newcommand{\proc}{\texttt}
\newcommand{\rankop}{\proc{rank}\xspace}
\newcommand{\selectop}{\proc{select}\xspace}
\newcommand{\selectopnospace}{\proc{select}}
\newcommand{\rankopnospace}{\proc{rank}}
\newcommand{\var}{\texttt}

\newcommand{\newrev}[1]{#1}
\newcommand{\newmymarginpar}[1]{}
\newcommand{\newmymarginnote}[1]{}
\newcommand{\newnew}[1]{#1}

\newcommand{\rev}[1]{#1}
\newcommand{\mymarginpar}[1]{}
\newcommand{\mymarginnote}[1]{}
\newcommand{\new}[1]{#1}

\definecolor{magenta4}{rgb}{0.5625,0,0.5625}
\definecolor{green4}{rgb}{0,0.5625,0}
\definecolor{orange4}{rgb}{0.98,0.31,0.09}

\newcommand{\head}{\texttt{head}\xspace}
\newcommand{\spillptr}{\texttt{spill\_ptr}\xspace}
\newcommand{\spilllist}{\texttt{spill\_list}\xspace}
\newcommand{\spilllistnospace}{\texttt{spill\_list}}
\newcommand{\dist}{\texttt{dist}\xspace}
\newcommand{\wordsdata}{\texttt{words}\xspace}
\newcommand{\urldata}{\texttt{url}\xspace}
\newcommand{\dnadata}{\texttt{dna}\xspace}
\newcommand{\xmldata}{\texttt{xml}\xspace}
\newcommand{\wikidata}{\texttt{wiki}\xspace}
\newcommand{\logdata}{\texttt{log}\xspace}
\newcommand{\wordsstar}{\texttt{words*}\xspace}
\newcommand{\urlstar}{\texttt{url*}\xspace}
\newcommand{\dnastar}{\texttt{dna*}\xspace}
\newcommand{\xmlstar}{\texttt{xml*}\xspace}

\newcommand{\numtries}{eight\xspace}
\newcommand{\numdatasets}{six\xspace}

\newcommand{\bvfstspeedup}{1.58$\times$\xspace}
\newcommand{\bvcocospeedup}{1.12$\times$\xspace}
\newcommand{\bvmarisaspeedup}{1.42$\times$\xspace}
\newcommand{\csqfstspeedup}{1.58$\times$\xspace}
\newcommand{\csqcocospeedup}{1.13$\times$\xspace}
\newcommand{\csqmarisaspeedup}{1.42$\times$\xspace}
\newcommand{\csqfstspacesaving}{1.27$\times$\xspace}
\newcommand{\csqcocospacesaving}{1.31$\times$\xspace}
\newcommand{\csqmarisaspacesaving}{1.30$\times$\xspace}
\newcommand{\toplevelspacesavings}{1.3$\times$\xspace}
\newcommand{\fsstvsrepair}{1.12$\times$\xspace}

\newcommand{\fsstmaxspacesavings}{1.3$\times$\xspace}

\newcommand{\marisarecursionspace}{1.17$\times$\xspace}
\newcommand{\marisarecursiontime}{1.30$\times$\xspace}

\newcommand{\repair}{re-pair\xspace}
\newcommand{\fsst}{FSST\xspace}
\newcommand{\sizeshort}{Size}
\newcommand{\queryshort}{Query}

\newcommand{\buildtimeshort}{Build}
\newcommand{\bv}{bitvector\xspace}
\newcommand{\Bv}{Bitvector\xspace}
\newcommand{\Bvs}{Bitvectors\xspace}
\newcommand{\bvs}{bitvectors\xspace}
\newcommand{\islink}{\var{IsLink}\xspace}
\newcommand{\haschild}{\var{HasChild}\xspace}
\newcommand{\louds}{\var{LOUDS}\xspace}
\newcommand{\child}{\var{child}\xspace}
\newcommand{\parent}{\var{parent}\xspace}
\usetikzlibrary{patterns}
\newcommand{\defn}[1]{\textit{#1}}
\renewcommand{\defn}[1]       {{\textit{\textbf{\boldmath #1}}}}
\newcommand{\figref}[1]         {Figure~\ref{fig:#1}}

\newcommand{\tabref}[1]        {Table~\ref{tab:#1}}
\newcommand{\secref}[1]         {Section~\ref{sec:#1}}

\newcommand{\secreftwo}[2]      {Sections \ref{sec:#1} and~\ref{sec:#2}}

\pgfplotsset{
    discard if/.style 2 args={
        x filter/.append code={
            \edef\tempa{\thisrow{#1}}
            \edef\tempb{#2}
            \ifx\tempa\tempb
                
            \fi
        }
    },
    discard if not/.style 2 args={
        x filter/.append code={
            \edef\tempa{\thisrow{#1}}
            \edef\tempb{#2}
            \ifx\tempa\tempb
            \else
                
            \fi
        }
    },
}

\pgfplotsset{compat=1.8,
    /pgfplots/ybar legend/.style={
    /pgfplots/legend image code/.code={%
       \draw[##1,/tikz/.cd,yshift=-0.25em]
        (0cm,0cm) rectangle (3pt,0.8em);},
   },
}
\pgfdeclareplotmark{rotated triangle}{%
  \pgfpathmoveto{\pgfqpoint{0pt}{-\pgfplotmarksize}}%
  \pgfpathlineto{\pgfqpointpolar{30}{\pgfplotmarksize}}%
  \pgfpathlineto{\pgfqpointpolar{150}{\pgfplotmarksize}}%
  \pgfpathclose
  \pgfusepathqstroke
}

\pgfdeclareplotmark{rotated triangle*}{%
  \pgfpathmoveto{\pgfqpoint{0pt}{-\pgfplotmarksize}}%
  \pgfpathlineto{\pgfqpointpolar{30}{\pgfplotmarksize}}%
  \pgfpathlineto{\pgfqpointpolar{150}{\pgfplotmarksize}}%
  \pgfpathclose
  \pgfusepathqfillstroke
}

\definecolor{safe-black}{RGB}{0,0,0}
\definecolor{safe-olive}{RGB}{0,73,73}
\definecolor{safe-teal}{RGB}{0,146,146}
\definecolor{safe-pink}{RGB}{255,109,182}
\definecolor{safe-peach}{RGB}{255,160,110}
\definecolor{safe-plum}{RGB}{73,0,146}
\definecolor{safe-cerulean}{RGB}{0,109,219}
\definecolor{safe-lavender}{RGB}{182,109,255}
\definecolor{safe-sky}{RGB}{109,182,255}
\definecolor{safe-baby}{RGB}{182,219,255}
\definecolor{safe-brick}{RGB}{146,0,0}
\definecolor{safe-brown}{RGB}{146,73,0}
\definecolor{safe-orange}{RGB}{219,209,0}
\definecolor{safe-green}{RGB}{36,255,36}
\definecolor{safe-yellow}{RGB}{255,255,109}
\definecolor{safe-darkgreen}{RGB}{36,200,36}

\pgfplotscreateplotcyclelist{mylist}{%
    {safe-pink, mark=square*},
    {safe-cerulean, mark=triangle*},
    {safe-peach, mark=x,mark options={scale=1.6}},
    {blue, mark=star,mark options={scale=1.6}},
    }

    \pgfplotscreateplotcyclelist{mylistshapes}{%
    {safe-pink, mark=square*},
    {safe-cerulean, mark=triangle*},
    {safe-peach,mark=*},
    {blue, mark=diamond*},
    }

\makeatletter
\def\BState{\State\hskip-\ALG@thistlm}
\makeatother

\newcommand{\bvdesign}{$\text{C}_1$\xspace}
\newcommand{\compressionscheme}{$\text{C}_2$\xspace}

\newcommand{\csquared}{$\texttt{C}^2$\xspace}



%% file: edbt-macros-v30-n1.tex

\def\EDBTISSN{2367-2005}
\def\EDBTISBN{978-3-89318-106-3}
\setcopyright{none}
\copyrightyear{\copyright\ 2026 Copyright held by the owner/author(s).
  Published on \url{OpenProceedings.org} under ISBN \EDBTISBN, series ISSN \EDBTISSN.
  Distribution of this paper is permitted under the terms of the
  Creative Commons license CC-by-nc-nd 4.0}
\acmYear{2026}
\acmDOI{}
\acmISBN{}

\acmConference[EDBT '27]{Extending Database Technology}{6-9 April 2027}{Lille (France)}

\settopmatter{printacmref=false, printccs=false, printfolios=false}

\newsavebox{\ximagebox}
\newlength{\ximageheight}
\newsavebox{\xglyphbox}
\newlength{\xglyphheight}
\newcommand{\xbox}[1]%
  {\savebox{\ximagebox}{#1}%
  \settoheight{\ximageheight}{\usebox{\ximagebox}}%
  \savebox{\xglyphbox}{\color{white}\char32}%
  \settoheight{\xglyphheight}{\usebox{\xglyphbox}}%
  \raisebox{\ximageheight}[0pt][0pt]{\raisebox{-\xglyphheight}[0pt][0pt]{%
    \makebox[0pt][l]{\usebox{\xglyphbox}}}}%
    \usebox{\ximagebox}%
    \raisebox{0pt}[0pt][0pt]{\makebox[0pt][r]{\usebox{\xglyphbox}}}}


\newsavebox{\LogoBox}
\sbox{\LogoBox}{\includegraphics[height=1cm]{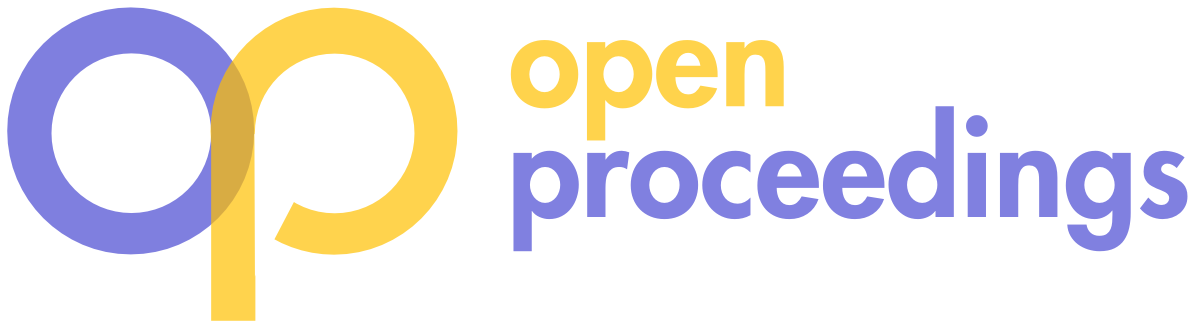}}

\fancypagestyle{firstpagestyle}{%
  \fancyhf{}%
  \fancyfoot{}%
  \fancyhead[L,C]{}
  \fancyhead[R]{\href{https://OpenProceedings.org/}{\raisebox{15pt}[0pt][0pt]{\xbox{\usebox{\LogoBox}}}}}
  }
  
\hypersetup{%
  pdfcopyright={CC-by-nc-nd},
  pdfpublisher={OpenProceedings.org -- University of Konstanz, Germany},
  pdfcreator={LaTeX with acmart, hyperref, and EDBT modifications},
  pdfvolumenum={30},
  pdfissuenum={1},
  pdfisbn={\EDBTISBN},
  pdfissn={\EDBTISSN},
  pdfdoi={10.48786/edbt.2027.xx}
}

%% file: new-abstract.tex
\begin{OPabstract}
  Succinct tries are powerful string dictionaries because of their low memory
  footprint and fast query performance. However, existing succinct trie
  implementations face two key challenges to spatial locality: 1) they incur
  unnecessary cache misses during queries, especially during trie navigation
  operations, and 2) they waste significant space when the data contains many
  unary paths. We propose \csquared, a set of two techniques: $\text{C}_1$
  introduces a more cache-friendly layout for the \bv underlying succinct tries,
  and $\text{C}_2$ compresses redundant unary paths. We thoroughly redesign
  three state-of-the-art succinct tries: FST, CoCo-trie, and Marisa, producing
  $\texttt{C}^2$-FST, $\texttt{C}^2$-CoCo, and $\texttt{C}^2$-Marisa.
  Experiments on six diverse datasets show that the \bvdesign optimization
  improves query performance by \bvfstspeedup, \bvcocospeedup, and \bvmarisaspeedup, respectively,
  compared to the original FST, CoCo-trie, and Marisa.  Furthermore, the
  $\text{C}_2$ optimization achieves a \toplevelspacesavings smaller memory
  footprint on average. The succinct tries optimized with both aspects of
  $\texttt{C}^2$ achieve better space-time tradeoffs than their original
  versions and other state-of-the-art succinct tries, while using significantly
  less space than non-succinct tries like ART and C-ART.
\end{OPabstract}


%% file: introduction.tex
\section{Introduction}\label{sec:intro}\mymarginpar{Meta 1, R1 O2, R3 O1}

Succinct data structures have become essential in applications like large-scale
databases, graph processing, and text indexing, especially in memory-constrained
scenarios~\cite{sds_application}. They are string-data representations that use
close to the theoretical minimum number of bits while still supporting efficient
queries~\cite{sds_overview, munro2018succinct}. We focus on \defn{succinct
  tries}~\cite{surf, pdt, coco_0, coco_1, marisa}, which match or exceed the
performance of regular tries \rev{with orders of magnitude less memory.}

Succinct tries separate the trie \defn{topology} from the string \defn{data},
encoding the topology in a compact bitvector that uses only about 2 bits per
node on average compared to at least 128 bits for two pointers in traditional
tries~\cite{jacobson1989space}. Navigation requires auxiliary index operations
(e.g., ``\rankop'' and ``\selectop'') on top of this bitvector.

\input{latexfigs/cache_motivation}

\newmymarginpar{Meta 1, R1 O2, R3 O1/D1}\para{Challenges to locality}
  Existing \bv-based succinct tries exhibit suboptimal locality in two respects:
  1) the \emph{bitvector layout} in the topology, and 2) the \newmymarginpar{R2
    O3/D6}\newrev{\defn{unary paths}, or maximal trie paths of at least 2 nodes
    in which every node except the last has exactly one child}, in the data.
  For example, in~\figref{louds-example}, paths 0-1-4 (``ca'') and 0-2-5
  (``su'') are internal unary paths, whereas path 7-14-18 (``ie'') is a suffix
  unary path.~\secreftwo{bv-design}{unary-path} detail both challenges.

\rev{\tabref{cache-motivation}} confirms locality challenges due to navigation
operations across state-of-the-art succinct tries. For example, a child
navigation in Marisa~\cite{marisa} incurs at least 3 cache misses to access the
topology \bv compared to 1 in a pointer-based trie to follow the pointer.

\input{latexfigs/unary-wiki-log}

\newnew{Additionally, practical datasets often include strings with long
  shared prefixes followed by diverse \new{dangling} suffixes, which naturally
  introduce large numbers of unary paths in trie-based representations.  For
  example, collections of hyperlinks frequently share common prefixes such as
  ``en.wikipedia.org/wiki/'', while differing only in their trailing titles. As
  a result, large portions of the trie structure degenerate into long chains of
  unary nodes.

  Without effective compression, unary paths can consume a dominant fraction of
  succinct-trie storage, increasing memory footprint and disrupting cache
  locality.  \newmymarginpar{R1 O1}For example, as shown
  in~\tabref{unary-wiki-log}, the majority of suffix regions in realistic
  datasets consist of long, redundant unary chains in the two largest tested
  inputs. In the \wikidata\ dataset, 84.1\% of all branch edges correspond to
  compressible unary paths, with an average compressible length of 271
  characters and a maximum compressible length of over 8K characters. Similarly,
  the \logdata\ dataset contains 78.8\% compressible unary paths.}

\para{Optimizing for locality in \bv-based succinct tries} \emph{We introduce
  \csquared, a set of techniques that improve locality in both the topology and
  the data of succinct tries.} \newmymarginpar{Meta 3, R1 O3, R3 O2}\newrev{We target two operations: the \defn{existence query} (checking
  whether a key exists) and the \defn{range query} (returning a range of keys).}

The first ``C'' of~\csquared improves cache locality of the \selectop index
in succinct tries. We introduce the \defn{functional index}, a novel
\selectop index layout that maps straightforwardly to bit positions, enabling
both \rankop and \selectop indexes to be inlined in the bit sequence.
\mymarginpar{Meta 4}\rev{Using the functional index, we cut random accesses to
  two per child navigation in LOUDS-based succinct tries, which include many
  state-of-the-art tries such as the FST~\cite{surf},
  CoCo-trie~\cite{coco_0}, and Marisa~\cite{marisa}. \tabref{cache-motivation}
  and \secref{experiments} show that this optimization reduces cache misses by
  19\% on average with at most $2\%$ overall space overhead.}

\mymarginpar{R3 O2, Meta 4}\newrev{The second ``C'' of~\csquared is an
  \defn{adaptive compression} algorithm that gives \emph{any} succinct trie
  access to compression in the ``tail container,'' which stores the remaining
  string suffixes at the end of trie paths. Building on
  ``\repair''~\cite{larsson2000off} and ``recursion''~\cite{marisa}, two
  orthogonal techniques for path compression, we enable all
  tries to use both recursion and the Fast Static Symbol Table
  (\fsst)~\cite{fsst}, a lightweight scheme that achieves good compression
  ratios with fast build times. Although recursion and compressed tail
  containers appear in prior work, this is the first generalization of
  recursion across tail containers, exposing space-time tradeoffs for any
  succinct trie.}


\input{latexfigs/intro_results_summary}
\para{Contributions} Our primary contributions are as follows:
\begin{itemize}[leftmargin=*,noitemsep]
\item We redesign the succinct-trie \rev{topology} with a \emph{cache-conscious}
  functional index that inlines both \rankop and \selectop indexes, minimizing
  cache misses during trie navigation.
\item \rev{We propose an \emph{adaptive compression scheme} that selects the
    best compression strategy per dataset to balance space and query
    performance.} \newrev{Our experiments indicate that the chosen strategy improves
  space efficiency by \toplevelspacesavings on average with similar
   (within $1.05\times$) query performance compared to the original succinct trie versions.} Furthermore, increasing compression aggressiveness exposes further space-time tradeoffs.
  \newrev{For example, on Marisa, increasing ``recursion'' for more compression yields
  \marisarecursionspace space savings, but \marisarecursiontime slower queries.}
\item We integrate these techniques into three state-of-the-art succinct tries
  --- FST~\cite{surf}, CoCo-trie~\cite{coco_0, coco_1}, and Marisa~\cite{marisa}
  ---
  improving query
  performance on the whole (with both \bvdesign and \compressionscheme
  optimizations) by \csqfstspeedup, \csqcocospeedup, and \csqmarisaspeedup,
  respectively.
\end{itemize}

\para{Results summary} Overall, \csquared-optimized
tries improve cache locality and performance over their original counterparts with negligible (<4\%) additional space overhead. 
\csquared-FST achieves better space-time performance than FST on all datasets. 
\rev{ 
    \csquared-CoCo outperforms CoCo' by \csqcocospeedup in query latency and is \csqcocospacesaving smaller in space, on  average across the six datasets
        \footnote{The original CoCo-trie was designed for prefix-only datasets and fails to build on the original datasets. 
        We implemented CoCo' by integrating CoCo-trie's topology with \csquared-CoCo.~\secref{experiments} provides the full details and shows that CoCo and CoCo' achieve almost identical performance and space usage on the prefix-only datasets.
        }. 
    \csquared-Marisa improves Marisa's query performance by \csqmarisaspeedup with similar memory consumption, and is generally the strongest \csquared index.
}

\figref{intro-results-summary} \rev{shows the query latency and size of the
  different tries} on the \rev{\wikidata} and \rev{\logdata} datasets.
\csquared-Marisa dominates all other succinct tries on query latency with the
lowest (or within a few percent of the lowest) space usage. \rev{ART and
C-ART, two non-succinct tries, achieve lower query latency at the cost of
higher space usage.}

\rev{These space-time improvements stem from the cache-locality gains of
  \csquared, as demonstrated by the cache-miss counts
  in~\tabref{cache-motivation}.}


%% file: latexfigs/cache_motivation.tex
\begin{table}[t]
  \caption{ Average number of LLC misses per query on the two largest datasets.
    FST~\cite{surf}, CoCo-trie~\cite{coco_1} and Marisa~\cite{marisa} are
    \bv-based succinct tries, while C-ART~\cite{hybrid_index} is a pointer-based
    compact trie.  Tries prefixed with \csquared- are our cache-optimized
    versions.  }
  \resizebox{\columnwidth}{!}{%
    \begin{tabular}{llllllll}
        \toprule
        \textit{Dataset}
        & \textit{FST}
        & \textit{\begin{tabular}{@{}l@{}}\csquared-\\FST\end{tabular}}
        & \textit{CoCo'}
        & \textit{\begin{tabular}{@{}l@{}}\csquared-\\CoCo\end{tabular}}
        & \textit{Marisa}
        & \textit{\begin{tabular}{@{}l@{}}\csquared-\\Marisa\end{tabular}}
        & \textit{C-ART} \\
        \midrule

        \wikidata~\cite{wiki-dataset}
        & 82
        & 78
        & 66
        & 49
        & 56
        & 58
        & \textbf{23} \\

        \logdata~\cite{log-dataset}
        & 235
        & 230
        & 57
        & \textbf{45}
        & 58
        & 62
        & 117 \\

        \bottomrule
    \end{tabular}
}
\label{tab:cache-motivation}
\end{table}

%% file: latexfigs/unary-wiki-log.tex
\begin{table}[t]

\caption{
Unary-path statistics for the two largest datasets.
\#Branch denotes the number of branching edges. Given a path length $\ell$, we
report the percentage of paths with lengths
 $\ell=1$, $1<\ell \leq 3$, and  
$\ell> 3$.
$\ell_{AVG}$ and $\ell_{MAX}$ denote the average and maximum
lengths of compressible unary paths. 
}
\centering
\small
\begin{tabular}{lrrrrrr}
\toprule
\textit{Dataset}
& \textit{\# Branch}
& \textit{$\ell = 1$}
& \textit{$1 < \ell \leq 3$}
& \textit{$\ell > 3$}
& $\ell_{AVG}$
& $\ell_{MAX}$ \\
\midrule

\wikidata~\cite{wiki-dataset}
& 467K
& 15.5\%
& 0.4\%
& 84.1\%
& 271
& 8676 \\

\logdata~\cite{log-dataset}
& 7,284K
& 17.4\%
& 3.8\%
& 78.8\%
& 65
& 8196 \\

\bottomrule
\end{tabular}

\label{tab:unary-wiki-log}

\end{table}


%% file: latexfigs/intro_results_summary.tex
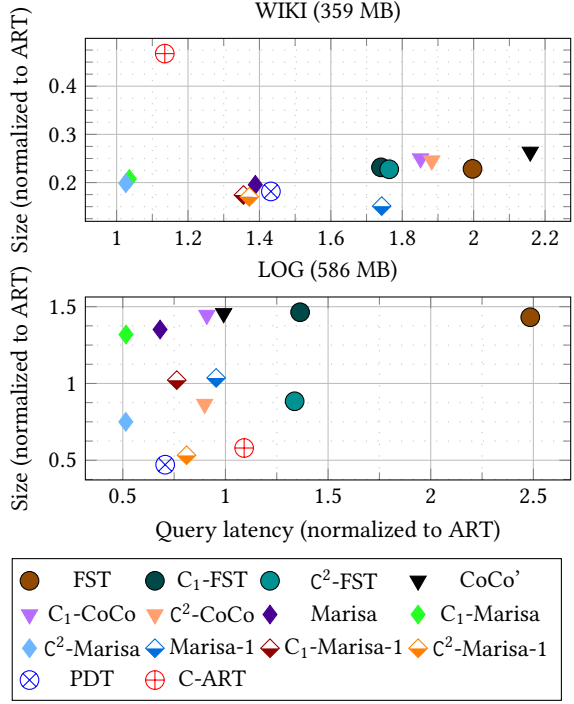
\begin{figure}[t]
    \pgfplotsset{every axis title/.append style={at={(0.5,0.95)}}}
    \centering
    \begin{tikzpicture}
    \begin{groupplot}[
        group style={group size=1 by 2},
        width=8cm, height=4cm,
        ymajorgrids=true,
        xmajorgrids=true,
        minor x tick num=1,
        xminorgrids=true,
        minor y tick num= 3,
        yminorgrids = true,
        minor grid style=loosely dotted,
        legend columns=4,
        legend entries={FST, \bvdesign-FST, \csquared-FST, CoCo', \bvdesign-CoCo, \csquared-CoCo, Marisa, \bvdesign-Marisa, \csquared-Marisa, Marisa-1, \bvdesign-Marisa-1, \csquared-Marisa-1, PDT, C-ART},
        legend to name=PerformanceLegend,
        scatter/classes={
            FST={mark=*, fill=safe-brown},
            C1-FST={mark=*, fill=safe-olive},
            C2-FST={mark=*, fill=safe-teal},
            CoCo'={mark=rotated triangle*, safe-black},
            C1-CoCo={mark=rotated triangle*, safe-lavender},
            C2-CoCo={mark=rotated triangle*, safe-peach},
            Marisa={mark=diamond*, safe-plum},
            C1-Marisa={mark=diamond*, safe-green},
            C2-Marisa={mark=diamond*, safe-sky},
            Marisa-1={mark=halfsquare*, safe-cerulean},
            C1-Marisa-1={mark=halfsquare*, safe-brick},
            C2-Marisa-1={mark=halfsquare*, orange},
            PDT={mark=otimes, blue},
            C-ART={mark=oplus, red}
        },
        scatter src=explicit symbolic
    ]

    \nextgroupplot[
        title={WIKI (359 MB)},
        ylabel={Size (normalized to ART)},
        legend to name=PerformanceLegend
    ]
    \addplot[scatter,only marks,mark size=3.5pt,point meta=explicit,scatter src=explicit symbolic]
    table[x=query,y=space,meta=trie]{tsvs/results-wiki-new-2.tsv};

    \nextgroupplot[
        title={LOG (586 MB)},
        xlabel={Query latency (normalized to ART)},
        ylabel={Size (normalized to ART)},
        legend to name=PerformanceLegend
    ]
    \addplot[scatter,only marks,mark size=3.5pt,point meta=explicit,scatter src=explicit symbolic]
    table[x=query,y=space,meta=trie]{tsvs/results-log-new-2.tsv};

    \end{groupplot}
    \end{tikzpicture}

    \begin{center}
    \ref{PerformanceLegend}
    \end{center}
    \caption{
    \rev{Space-time performance
        comparison on \wikidata~\cite{wiki-dataset} and
        \logdata~\cite{log-dataset} (normalized to ART). We use FSST as the tail
        container in all \csquared-tries. The prefixes "\bvdesign-"/"\csquared-"
        indicate different optimization strategies (bitvector redesign /
        bitvector redesign + unary path compression); the suffix "-1" means we
        enforce one recursion (see \secref{experiments}). Lower and to the left
        is better.}}
    \label{fig:intro-results-summary}
\end{figure}

%% file: prelim.tex
\section{Preliminaries}\label{sec:preliminaries}

We introduce the \rankop/\selectop primitives on \bvs and common succinct
tree-encoding schemes, review four state-of-the-art succinct tries this paper
builds upon, and summarize string compression schemes for the tail container.

\subsection{Rank/Select Primitives}

\Bvs are a fundamental building block of many succinct data structures. Succinct
tree navigation relies on the \rankop and \selectop primitives of \bvs, which
are defined as follows:
\begin{itemize}
  \item \proc{rank$_p$(i)}: return the number of pattern $p$'s in the first $i$ bits of the \bv.
  \item \proc{select$_p$(i)}: return the end of the $i$-th occurrence of pattern $p$ in the \bv.
\end{itemize}
Different \bvs support different patterns for \rankop and \selectop. Common
choices of $p$ include $1$, $0$ and $00$. For example, \proc{rank$_1$(i)}
returns the number of occurrences of ``1'' in \rev{positions} [0, i).


Index structures for \rankop and \selectop partition the \bv into fixed-sized
(e.g., 512-bit) \rev{\emph{basic blocks}} and store partial results so far for
each block from left-to-right~\cite{rank_select_1}. \rev{They also often include
  \emph{secondary indices} on smaller fixed-size blocks within each basic block
  for higher granularity and better performance.} These auxiliary structures
reduce \newmymarginpar{R2 O3/D6}\newrev{the worst-case cost} of \rankop and
\selectop to $O(1)$ with minimal (typically \rev{about} $5\%$) space overhead
\cite{surf,rank_select_0,rank_select_1,rank_select_2,rank_select_3,sds_overview,ds2i,optimized_rank,terark_bv}.

For the \rankop operation, the \defn{rank index} stores \mymarginpar{R2
  D14}\rev{\defn{samples}, or exact accumulative ranks before each block (e.g.,
  the number of 1s preceding each basic block), which enable fast exact
  responses to \rankop queries.} Answering a \rev{\proc{rank$_p$(i)} query}
involves combining the accumulative rank before the target block with the
(computed) rank in the remainder up to position \proc{i}.

\mymarginpar{Meta 1, R1 D1}\new{For example,~\figref{traditional-bv} illustrates
  a worked example of a \\ \proc{rank$_1$(12)} query on a bitvector where the rank
  index is built with block size 8. First, we look up the result of
  \proc{rank$_1$(8) = 4} in $O(1)$ time in the rank index because position 8 is
  the start of the basic block containing 12. Second, we calculate the remaining
  instances of the pattern in the range 8 - 12 in the bitvector. The combined
  result of the first and second step is 4 + 1 = 5, which is
  \proc{rank$_1$(12)}.}

\rev{The \selectop operation uses a \defn{select index} similar
  to the aforementioned rank index, but \selectop is more challenging than
  \rankop because a query may not know exactly which block is needed based on
  the input. For example, with 8-bit basic blocks, bit 12 is always in the
  second basic block, but the 12-th \emph{one bit} might be in the 100-th basic
  block!

  To support the \selectop operation, the \selectop index precomputes and stores
  samples,\mymarginpar{R2 D14} or exact results of \selectop queries at regular
  intervals (e.g., \proc{select$_p$(0)}, \proc{select$_p$(8)},
  \proc{select$_p$(16)}, \ldots). Answering a \selectop query involves finding
  the closest index entry to the left of the desired select query, then scanning
  sequentially using that as a starting point.}  \mymarginpar{Meta 1, R1
  D1}\new{For example,~\figref{traditional-bv} illustrates a worked example of a
  \proc{select$_1$(14)} query on a bitvector where the select index is built
  with block size 8. First, we look up the closest entry below, which is
  \proc{select$_1$(8)} = 17, in $O(1)$ time. Therefore, we know that the result
  of \proc{select$_1$(14)} must be to the right of \proc{select$_1$(8)}, so we
  scan left to right starting from position 17 in the bitvector until we have
  found 6 more bits (for 8 + 6 = 14 bits total).}

\begin{figure}[t]
    \centering
    \includegraphics[width=.9\columnwidth,page=4,trim=7cm 6cm 9cm 6cm,clip]{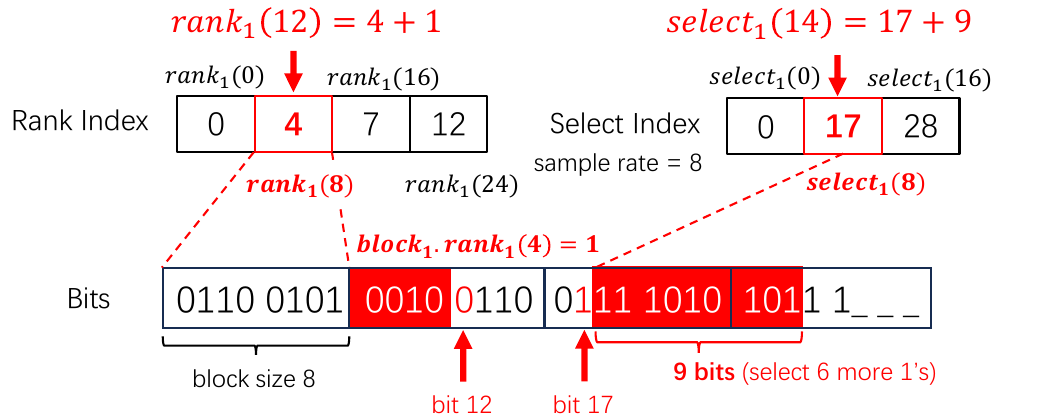}
    \caption{Examples of rank/select operations using the associated indices on
      top of a bitvector.}
    \label{fig:traditional-bv}
\end{figure}

\rev{\para{Succinct trie \bv structure} BP, LOUDS, and DFUDS are three popular
  \bv-based succinct tree-encoding schemes \cite{jacobson1989space,
    benoit2005representing}. They all use the topology bitvector to map each
  tree element (\rev{i.e.,} nodes, edges, internal nodes, leaf nodes) to a
  unique integer ID in the range of $[0, n-1]$, where $n$ is the number of such
  elements in the tree. Depending on the encoding scheme, the topology may
  support efficient tree operations \rev{(e.g., parent/child navigation)}
  \rev{via \bv operations (e.g., rank/select).} Succinct trees encoded using
  these schemes can be converted to tries by storing data associated with each
  trie edge in a separate array, which is indexed using the edge IDs~\cite{surf,
    marisa, pdt, coco_1}.

  Efficient bitvector operations require auxiliary data structures which
  separate the bitvector into two parts: the \defn{index}, which accelerates
  bitvector operations, and the \defn{bit sequence}, which stores the content of
  the bitvector (i.e. the trie topology). The index facilitates
  \rankopnospace/\selectop operations essential to succinct trie navigation.}

\subsection{LOUDS Encoding}\label{sec:louds}

Level-Order Unary Degree Sequence (LOUDS) is a
\newmymarginpar{R2 O1/D4}\newrev{classical example of} a \bv-based succinct-tree
encoding scheme~\cite{jacobson1989space}. Each tree node is encoded using the
bit sequence $1^{m}0$ (\rev{i.e.}, $m$ 1's followed by a 0), where $m$ is the
degree (\rev{i.e.}, the number of children) of the node. The encoding bits of
each node are concatenated in level order to form the encoding \bv \var{bv} of
the entire tree. With this encoding, each 1 bit corresponds to a parent-to-child
edge, whereas each 0 bit marks the end of a node. We say a position $i$ belongs
to a node $k$ if bit $i$ is part of node $k$'s encoding. The LOUDS encoding
efficiently supports the following tree operations:
\begin{itemize}
    \item \proc{NodeID(i) = bv.rank$_0$(i)}, the LOUDS ID of the node to which position $i$ belongs.
    \item \proc{EdgeID(i) = bv.rank$_1$(i)}, the LOUDS ID of the edge at position $i$ (Requirement: \proc{bv[i] == 1}).
    \item \proc{LeafID(i) = \rev{bv}\mymarginpar{Meta 5, R2 D17}.rank$_{00}$(i)}, the LOUDS ID of the leaf node to which position $i$ belongs.
    \item \proc{InnerID(i) = NodeID(i) - LeafID(i)}, the LOUDS ID of the inner node to which position $i$ belongs.
    \item \proc{Child(i)} = \proc{bv.select$_0$(bv.rank$_1$(i+1)) + 1}, the position of the child node connected to the edge at position $i$.
    \item \proc{Parent(i) = bv.select$_1$(bv.rank$_0$(i))}, the position of the parent of the node to which position $i$ belongs.
\end{itemize}

As illustrated in~\figref{louds-example}, encoding a trie using LOUDS involves
the following steps: First, the \emph{topology} of the trie is encoded using the
LOUDS \bv; Second, all the trie labels are concatenated in level order into a
separate array. At query time, trie labels are uniquely indexed by edge IDs
(because they are associated with trie edges), whereas keys are indexed by leaf
IDs.

\mymarginpar{Meta 1, R1 D1}\rev{Next, we will present an example of how to
  traverse a LOUDS-encoded trie using the trie
  in~\figref{louds-example}. Suppose we are navigating to the second child of
  the root (i.e., from node 0 to node 2). We need to follow the second edge of
  the root node, which is at position 1. From the aforementioned formula,
  \proc{Child(1)} = \proc{bv.select$_0$(bv.rank$_1$(2))+1} =
  \proc{bv.select$_0$(2)+1} = 6. Indeed, if we look at the LOUDS bitvector
  in~\figref{louds-example}, the second child of the root starts at position 6.}

Balanced Parentheses (BP)~\cite{jacobson1989space} and Depth-First Unary Degree
Sequence (DFUDS)~\cite{benoit2005representing} are two popular alternatives to
LOUDS with the same two-bits-per-node overhead, but they traverse the tree in
different orders. Compared with LOUDS, they \rev{support}\mymarginpar{Meta 5, R1
  D9} a broader range of tree operations, including computing subtree sizes or
counting the number of leaf nodes to the left/right of a specific node. However,
\rev{their functionality} comes at the cost of slower parent/child navigation
\rev{compared to LOUDS due to their greater} complexity.

\begin{figure}[t]
    \centering
    \includegraphics[width=.48\textwidth,page=1,trim=7cm 5cm 7cm 4.5cm,clip]{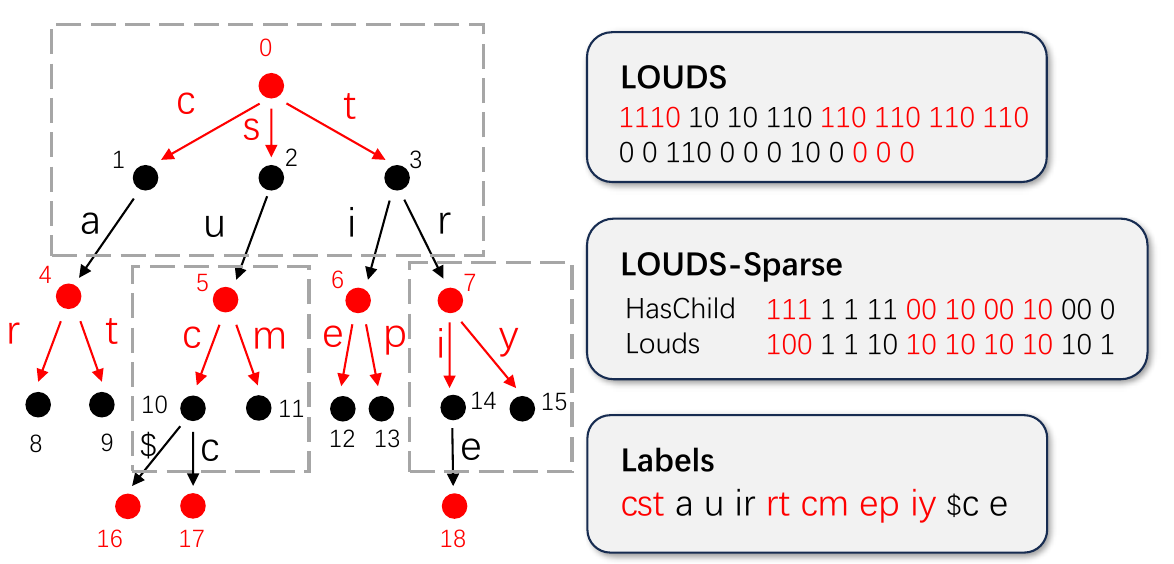}

    \caption{LOUDS/LOUDS-Sparse encoding example. Stored keys: car, cat, suc, succ, sum, tie, tip, trie, try.}
    \label{fig:louds-example}
  \end{figure}

\subsection{State-of-the-art Succinct Tries}\label{sec:existing-succinct-tries}

\para{Fast Succinct Trie} The Fast Succinct Trie (FST)~\cite{surf} is a static
succinct trie that combines the benefits of LOUDS-Sparse and LOUDS-Dense by
leveraging both formats in different levels of the trie. Specifically, FST
encodes the top levels of the trie, where nodes tend to have more branches and
are more likely to be queried, using LOUDS-Dense, a representation that is fast
and space efficient for nodes with dense branches. The rest of the trie, which
tends to be colder and of lower average degree, is encoded with the more compact
LOUDS-Sparse. \new{Due to space limitations, we omit the details of
  LOUDS-Dense, 
  but we refer the reader to the original FST paper~\cite{surf} for the full
  details.}

LOUDS-Sparse uses 2 bits to encode each trie edge. The \\\proc{HasChild} bit of
each edge indicates if this edge points to an internal node (1) or a leaf node
(0). The \proc{Louds} bit of each edge indicates if this edge points to the
first child of current node (1) or not (0), which is used to determine node
boundaries. These two sets of encoding bits are respectively concatenated in
level order to form two \bvs, \proc{HasChild} and \proc{Louds}. With
LOUDS-Sparse, we have
\begin{itemize}
    \item \proc{Child(i) = Louds.select$_1$(HasChild.rank$_1$(i+1)+1)}.
    \item \proc{Parent(i) = HasChild.select$_1$(Louds.rank$_1$(i+1)-1)}.
    \item \proc{LeafID(i) = HasChild.rank$_0$(i)}.
\end{itemize}

\rev{Let us consider how to perform trie navigation in a LOUDS-sparse-encoded
  trie using the example in~\figref{louds-example}. Suppose we are navigating
  again to the second child of the root (also at position 1). \mymarginpar{Meta
    1, R1 W1, R1 D1}From above, \proc{Child(1) =
    Louds.select$_1$(HasChild.rank$_1$(2)\\+1) = Louds.select$_1$(3)} = 4. Indeed,
  the second child of the root starts at position 4 in the \proc{Louds}
  bitvector in the LOUDS-Sparse encoding.}

\begin{figure}[t]
    \centering
    \includegraphics[width=\columnwidth,page=2,trim=7cm 5cm 9cm 3.5cm,clip]{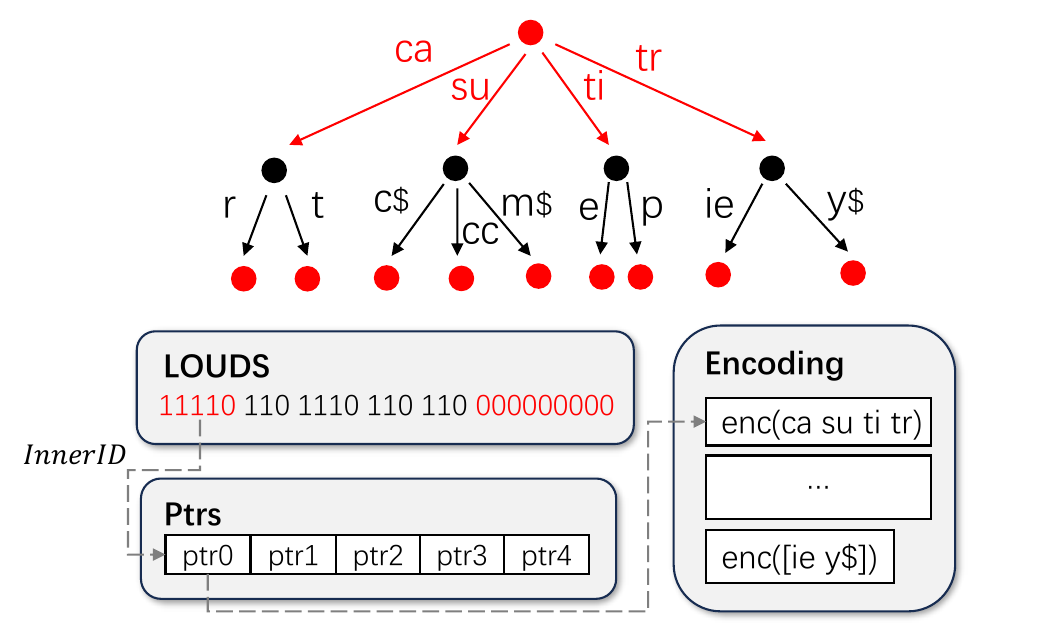}
    \caption{An example of a CoCo-trie obtained by collapsing sub-tries
      surrounded by dashed boxes in~\figref{louds-example}.}
    \label{fig:coco-example}
\end{figure}

\para{CoCo-Trie} The CoCo-trie \cite{coco_0, coco_1} is a recent
state-of-the-art static succinct trie that achieves both low memory consumption
and fast query performance. A CoCo-trie is constructed by first creating a
regular trie $T$ and then adaptively selecting sub-tries of $T$ to collapse into
macro-nodes. Keys in each collapsed macro-node are transformed to an increasing
sequence of integer codes which is then encoded using a select pool of succinct
integer encoding algorithms. The optimal set of sub-tries to collapse is
computed using a \rev{bottom-up} dynamic-programming
algorithm. 

\mymarginpar{Meta 1, R1 W1, R1 D1}\rev{~\figref{coco-example} illustrates the
  CoCo-trie encoding of the example trie in~\figref{louds-example}. Since the
  CoCo-trie uses LOUDS encoding, its trie navigation is done through
  \rankop/\selectop operations on the bitvector as described in~\secref{louds}.}

\para{Marisa Trie}\label{sec:marisa} The Marisa trie~\cite{marisa} is a
LOUDS-encoded Patricia trie~\cite{partricia} that contracts unary paths into
single concatenated labels. Unlike Patricia, Marisa improves space efficiency by
recursively storing unary paths in multiple tries.

\new{The upper subgraph of \figref{marisa-example} shows an example of Patricia,
  where unary paths are contracted into a single node. To convert the trie to
  LOUDS succinct form, two bitvectors are needed: the topology bitvector
  \texttt{LOUDS} (\secref{louds}), and an additional bitvector (\texttt{IsLink})
  which indicates if each edge is a regular edge (0) or a contracted edge
  (1). Each contracted edge is uniquely indexed by a \var{LinkID} determined by
  a \rankop operation on the \var{IsLink} \bv. The \texttt{Links} vector stores
  \defn{links}, or references, to the concatenated labels on the contracted
  edge. For example, bit 3 in \texttt{LOUDS} corresponds to bit 2 in
  \texttt{IsLink} and link 0 in \texttt{Links} (``omp'').}

\mymarginpar{R1 D7}\new{Marisa improves on Patricia by recursively storing unary
  paths in secondary tries.} The recursion \rev{continues} until the number of
tries reaches a \rev{preset limit}, at which point the outstanding unary paths
are sorted, deduplicated and moved to a \rev{simple ``sorted'' tail container
  (\secref{text-compression}).} \new{The lower subgraph of
  \figref{marisa-example} shows a worked example of a Marisa trie with one
  recursion. For readability, links between tries are expressed as leaf node IDs
  in this example. Because the second trie only needs to support random access
  but not lookup operations, the keys are stored in reversed form so that they
  can be efficiently retrieved by reading the labels on a bottom-up path. For
  example, to retrieve link 4's corresponding key, we start from leaf 4 of the
  second trie, concatenate all the (reversed) labels on the leaf-to-root path
  (`o' and `mp'), and obtain `omp'. In practice, links are implemented as leaf
  node positions to avoid an additional \selectop operation. Finally, to
  mitigate the cost of link tracing during lookups,
  Marisa uses a small piece of memory to cache frequently-taken paths.} 

\begin{figure}[t]
    \centering
    \includegraphics[width=\columnwidth,page=11,trim=8.2cm 5.4cm 10cm 3.9cm,clip]{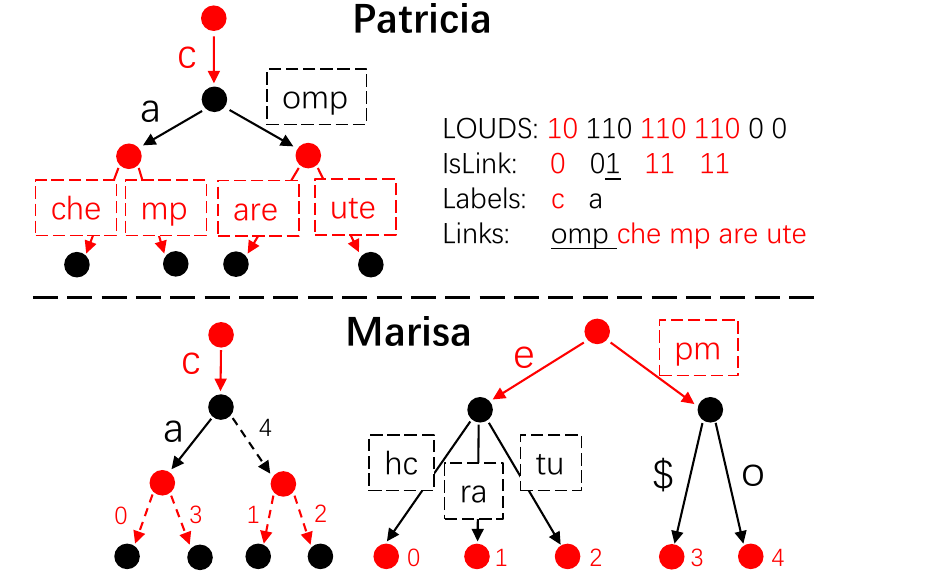}
    \caption{
      \rev{Examples of Patricia (top) and Marisa (bottom) with one
        recursion. Keys stored: cache, camp, compare, compute. Unary paths in
        the first Marisa trie are stored in the second Marisa trie in reversed
        form: ehc, era, etu, pm, pmo. Links in the first Marisa trie are
        indicated using dashed arrows, and the numbers alongside them indicate
        their values (which in this example are leaf node IDs in the second
        Marisa trie).}}
    \label{fig:marisa-example}
\end{figure}

\new{The number of recursions in Marisa exposes a \emph{space-time tradeoff}, as
  more recursions further compress the data at the cost of degraded search
  performance and higher build times~\cite{marisa-code}.}

\para{Path Decomposed Trie}
The Path Decomposed Trie (PDT) \cite{pdt, pdt_theory, dynamic_pdt} is a
DFUDS-encoded succinct trie. A PDT is obtained by transforming a regular trie
through path decomposition \cite{path_decomp}: At each step, a root-to-leaf path
is selected and contracted into a single node, and the subtrees dangling on the
original path are converted to its child nodes. The procedure \rev{recurses} 
 on the dangling subtrees.

 The PDT achieves both significant memory reduction and competitive query
 performance by compressing the contracted paths with an approximate version of
 the \defn{\repair} \mymarginpar{Meta 1, R1 W1}\rev{text compression
   algorithm~\cite{larsson2000off, repair}, which we will summarize in the next
   subsection.}

\new{\subsection{Text Compression
     Algorithms}\label{sec:text-compression} \mymarginpar{R1
     D7}\figref{repair-example} presents a worked example of how to compress a
   set of strings with the sorted tail container, \repair algorithm, and
   \fsst. Next, we summarize these algorithms.

  \para{Sorted tail container} \mymarginpar{Meta 1}The sorted tail container is
  the string container for the last Marisa trie on a recursion chain. \rev{It
    overlaps two keys if one is a suffix of the other, and is particularly space
    efficient for datasets with short and repetitive keys. To efficiently detect
    possibilities of overlapping, it needs to reverse and sort all keys. Once
    the strings have been sorted, construction takes linear time.} \new{The link
    array which is used to retrieve keys from the container should also preserve
    the original key ordering.}

  \new{\para{Dictionary compression} Dictionary compression encodes common text patterns as small integer
    codes, exemplified by the famous LZW algorithm~\cite{lzw}. The mapping
    between patterns and integers is known as the \defn{dictionary}, and finding
    a good dictionary is known to be a key challenge to such
    algorithms~\cite{fsst}. In our work, we focus on two related algorithms that
    support fast random access: re-pair~\cite{larsson2000off} and
    FSST~\cite{fsst}.}

  \para{Exact and approximate \repair} At a high level, the original \repair
  algorithm~\cite{larsson2000off} replaces the most frequently-occurring
  character pair with a new special character, proceeding in rounds. \rev{For
    example, using special characters $\alpha$, $\beta$, $\gamma$ and $\delta$,
    \repair may compress the string aabbaabb as $\alpha$bb$\alpha$bb ($\alpha$ =
    aa) in the first pass, $\alpha\beta\alpha\beta$ ($\beta$ = bb) in the second
    pass, $\gamma\gamma$ ($\gamma = \alpha\beta$) in the third pass and $\delta$
    ($\delta = \gamma\gamma$) in the fourth pass.} \new{The dictionary is
    expressed as a set of recursive rules of pairing, which can be flattened to
    support constant-time decoding~\cite{pdt}, with a small space
    overhead. \mymarginpar{R1 D4}Prior work shows that \repair compresses a
    sequence $T$ of length $n$ over an alphabet of size $\sigma$ and $k$-th
    order entropy $H_k$ to $O(H_k)$ bits.~\cite{navarro2008re,
      ochoa2018repair}.}

\begin{figure}[t]
    \centering
        \includegraphics[width=\columnwidth,page=14,trim=7.5cm 11.5cm 11.5cm 2cm,clip]{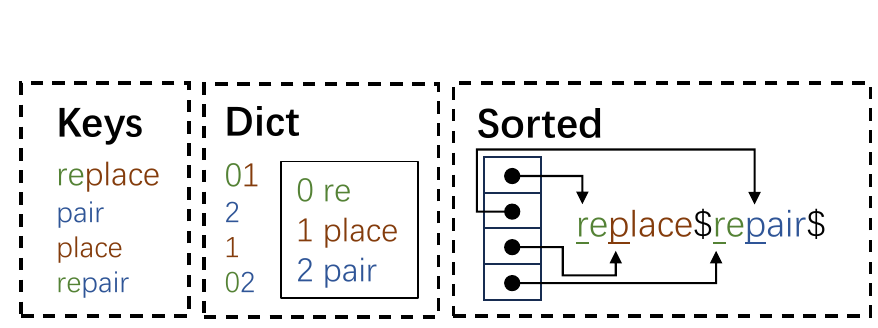}
        \caption{
        \rev{Examples of re-pair, FSST and sorted tail containers. Re-pair and FSST are both dictionary encoding methods (denoted with \texttt{Dict}) that replaces common text patterns with small integer codes. The sorted tail container (denoted with \texttt{Sorted}) overlaps suffix keys.}}
        \label{fig:repair-example}
  \end{figure}


  In practice, the PDT uses an \emph{approximate} version of the \repair
   algorithm~\cite{repair} that accelerates build time by orders of magnitude compared to exact \repair by identifying and
   replacing the top $k$ 
   most frequently-seen character pairs in
   each round, rather than one at a time. To our knowledge, there is no
   theoretical bound on how far the compression ratio of approximate
   \repair is from exact \repair's, as it depends on the dataset.

   \para{Fast Static Symbol Table} The Fast Static
   Symbol \\Table (\fsst)~\cite{fsst} is similar to approximate \repair, but
   optimizes for \emph{lightweight} compression (with linear build
   times). Specifically, it restricts the size of the symbol table (i.e., the
   list of special characters for replacement) to 256 entries and expands the
   size of each replaced entry to substrings of 8 bytes (rather than two
   characters in the original approximate \repair). \fsst further optimizes for
   build time by first encoding
   a small sample (about 16KB) of the original dataset, and then using the
   resulting symbol table to encode the entire dataset. FSST implements
   hardware-specific optimizations, achieving extremely high compression and
   decompression throughput. In practice, \fsst reduces the compression time by
   over an order of magnitude compared to approximate \repair, while achieving
   similar compression ratios. However, \newmymarginpar{R2 O1/D4}\newrev{no
     theoretical space bound has been proved for FSST}.}

%% file: bv-design.tex
\section{The First ``C'': Cache-conscious \Bv Design} \label{sec:bv-design}

In this section, we present the first "C" of \csquared:
\underline{C}ache-conscious bitvector design (denoted with \bvdesign). First,
inspired by prior work\\~\cite{optimized_rank, terark_bv}, we apply an
array-of-structs reorganization to bitvectors to better exploit the access
pattern of \rankop queries. Second, we eliminate cache misses from reading the
\selectop index using our novel \emph{functional index}. Third, to handle both
sparse and dense patterns, \rev{we further optimize the existing technique of
  overflowing \selectop indexes and reduce cache misses by storing all data
  in-place}. Finally, we apply the above techniques to three state-of-the-art
\mymarginpar{Meta 4}\rev{LOUDS-based} succinct tries~\cite{surf, coco_0,
  marisa}.

\rev{We adopt two design principles: (1) optimization priority of
  cache $>$ branch $>$ arithmetic operations~\cite{rank_select_1}, and (2)
  trading space for performance. Since trie topologies are usually much smaller
  than trie labels, a small space overhead is worthwhile for increased query
  speed.}

\para{Locality issues in \bvs} \newmymarginpar{Meta 1, R1 O2}\newrev{The main
  challenge to locality during trie navigation arises from storing the index and
  \bv separately~\cite{rank_select_0, rank_select_1, rank_select_2,
    rank_select_3, sds_overview, ds2i, surf}, since both must be accessed during
  \rankop and \selectop. Packing the \rankop index with the bit sequence
  improves locality for \rankop~\cite{optimized_rank, terark_bv}, but the same
  approach does not work for \selectop. In most cases, the input to \rankop maps
  linearly to a position in the bit sequence; \newmymarginpar{R2 O3/D6}the input
  to \selectop, however, is often an intermediate value (e.g.\ the output of
  \rankop) with no straightforward mapping to a bit position.}

\begin{figure}[t]
    \centering
    \begin{minipage}[h]{.19\textwidth}
    \includegraphics[width=\textwidth,page=5,trim=13cm 6.5cm 14cm 6cm,clip]{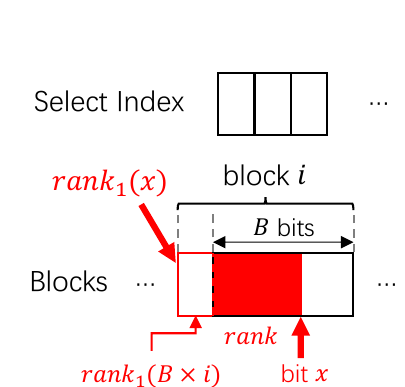}
        \caption{Interleaving.}
        \label{fig:interleaving}
    \end{minipage}
    \hspace{0.1cm}
    \begin{minipage}[h]{.26\textwidth}
        \includegraphics[width=\textwidth,page=7,trim=12cm 5cm 10cm 6cm,clip]{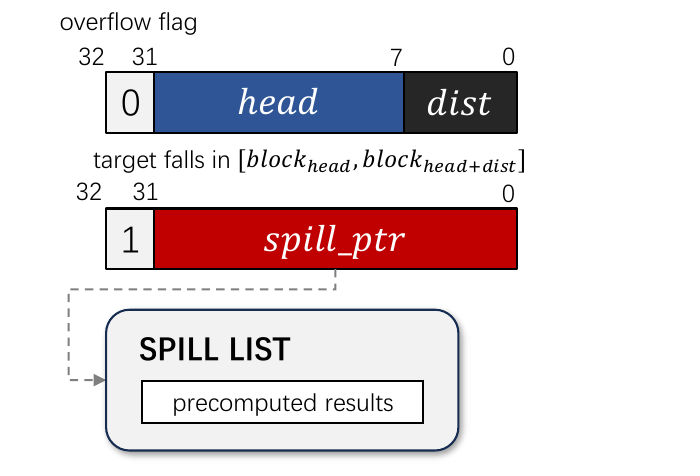}
        \caption{Select index overflow.}
        \label{fig:overflow}
    \end{minipage}
\end{figure}

\subsection{Array-of-Struct Reordering}\label{sec:interleaving}

Most traditional \bvs are organized in a struct-of-array manner. During \rankop
queries, each query accesses the target block immediately after its \rankop
index. To better exploit this locality, \rankop indexes can be \mymarginpar{Meta
  5, R2 D13}\rev{\defn{interleaved}, or stored in an array-of-structs manner
  with the index and bitvector in one contiguous memory allocation, as
  illustrated in \figref{interleaving}.}

Similarly, for tries encoded with multiple \bvs, we interleave blocks of
different \bvs that are often jointly accessed at nearby bit positions. For
example, in LOUDS-Sparse, we pack the blocks of \haschild and \louds, as they
are aligned with each other and are typically accessed together. \rev{Similarly,
  we can interleave secondary \rankop indexes~\cite{optimized_rank,
    terark_bv, sds_overview, rank_select_3, marisa}.} 

We restrict the size of each \rankop index \rev{element} to 32 bits \rev{for
  better space efficiency without loss of applicability to larger datasets}. To
support \bvs larger than $2^{32}$ bits, previous work shows how to first
\emph{partition} the data into \rev{connected} sub-tries that individually fit
in the size limit~\cite{hybrid_surf}.

\subsection{Functional Index}\label{sec:functional-index}
  \begin{figure}[t]
    \centering
    \includegraphics[width=\columnwidth,page=13,trim=3cm 8.4cm 8cm 3.5cm,clip]{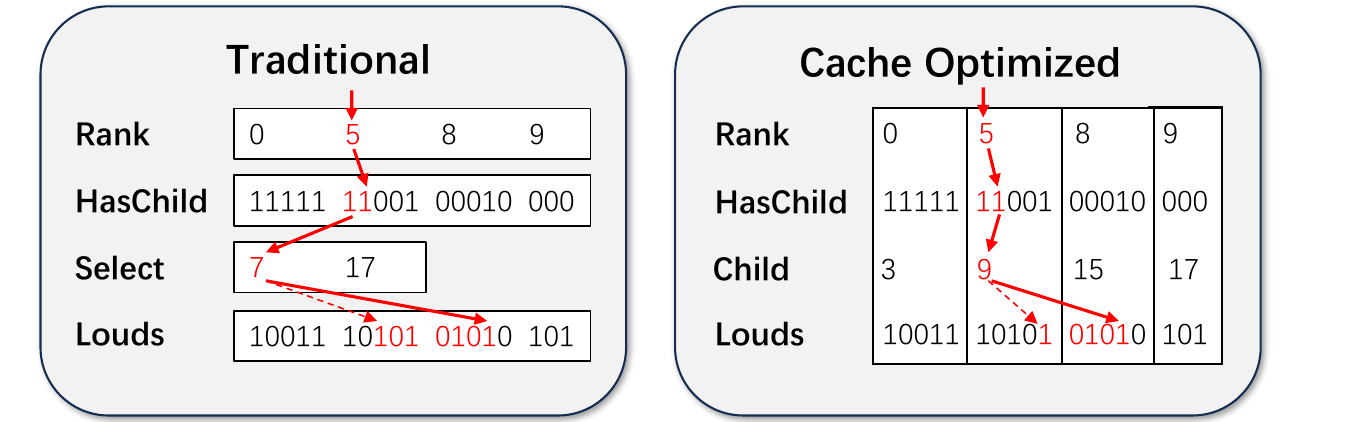}
    \caption{ \rev{Functional index example for LOUDS-Sparse's \texttt{Child(x)
          = Louds.select}$_1$\texttt{(HasChild.rank}$_1$\texttt{(x + 1) +
          1)}. The block size and the sampling rate for \selectop (left) are
        both 5. The third row on the right (\texttt{Child}) caches the results
        of \texttt{Child(x)} at the start of block (i.e. \texttt{Child(0)},
        \texttt{Child(5), etc.}).}}
    \label{fig:functional-index}
\end{figure}

To improve spatial locality for \selectop, we introduce the \defn{functional
  index} \mymarginpar{R2 D14}\rev{that samples results of the navigation
  function (e.g., \texttt{Child(x)}) rather than intermediate \selectop values.}
For example, to support child navigation in LOUDS (where
$\texttt{Child(x)} =
\selectopnospace_0\texttt{(}\rankopnospace_1\texttt{(x+1))+1}$), we sample
\texttt{Child(x)}, instead of \\$\selectopnospace_0$\texttt{(x)}, at regular intervals of
$x$. Since the sampled function is monotonically non-decreasing, the sampling is
unambiguous. \rev{We then interleave these samples with blocks for spatial
  locality.} Other succinct trie operations (e.g., parent navigation) can be
optimized using similar approaches.

\mymarginpar{R1 D2}\new{\figref{functional-index} shows a worked example of
  computing \texttt{Child(6)} using traditional and functional indexes. An
  implementation using the traditional index first computes
  \texttt{HasChild.\rankopnospace$_1$(6+1)=7} (by accessing \texttt{HasChild}'s
  rank index and the \texttt{HasChild} bitvector) and then obtains
  \texttt{Louds.select$_1$(7+1)=13} (by accessing \texttt{Louds}' select index
  and the \texttt{Louds} bitvector itself), which easily incurs 4 cache
  misses. In contrast, using the cache-optimized functional index, we first
  access block 1 (which bit 6 falls into) and obtain
  \texttt{HasChild}.\rankopnospace$_1$(6+1) = 7 (using block 1's rank
  index). Then, we read block 1's functional index for \texttt{Child}, which
  stores the value of \texttt{Louds}.\selectopnospace$_1$(5+1) = 9. Since we
  need to compute \selectopnospace$_1$(7+1), we select two more ones starting
  from bit 9 and end up at bit position 13.

   Assuming the two accessed blocks belong to different cache lines, we incur
   only two cache misses in total.} \rev{In practice, since} trie navigation
 operations are often executed in a pipelined manner (\textit{i.e.}, the output
 of one operation is immediately used as the input \rev{of} the next operation),
 this optimization essentially cuts cache misses by $4\times$ in long pipelines.

 \newrev{\para{Applicability beyond LOUDS}
   Although we develop and evaluate the functional index in the context of
   LOUDS-based tries, the core idea applies to any succinct-tree encoding whose
   navigation functions are monotonically non-decreasing compositions of \rankop
   and \selectop. The key insight behind the functional index is that
   \texttt{Child(x)} is a function of the position \texttt{x}, not of an
   intermediate bitvector value, so sampling it at regular intervals of
   \texttt{x} aligns the index with the bit sequence and enables
   interleaving. This property holds in BP and DFUDS as well, since the
   navigation function's input and output both increase monotonically with the
   traversal order of the encoding.

\newmymarginpar{R2 O4/D7, R3 O3/D3}
  In BP, child navigation reduces to a \texttt{findclose} operation (i.e.,
  locating the matching close parens for an open parens),
  which is implemented via \rankop/\selectop on the parenthesis \bv. The
  functional index directly applies because \texttt{findclose(x)} is
  monotonically non-decreasing in \texttt{x}, so it can be directly sampled at
  regular intervals of \texttt{x} and interleaved with the bit
  blocks. Similarly, DFUDS navigation starts with the \texttt{Child(x, 1)}
  operation, which returns the first child of the node at position
  \texttt{x}. This first-child operation admits a functional index in
  the standard way, with subsequent siblings located by scanning forward.
}

\subsection{Overflow of Select Index}\label{sec:overflow}

So far, we have assumed that the probe distance within the target interval is
small. But this is not always the case, especially for the CoCo-trie, which
sometimes contains nodes with very large fan-outs. For instance, when encoding
the \texttt{trec-terms}~\cite{trec_dataset} dataset in the original CoCo-trie
paper, the first 600k bits of the CoCo-trie topology contain only 36
\rev{zero}\mymarginpar{Meta 5, R2 D17} bits. With such extreme sparsity, regular
searching algorithms would be too slow.

A classic solution to this problem is to allow the \selectop index to overflow
\cite{ds2i, rank_select_3}. That is, the \selectop index should detect regions
with extremely low pattern density, and precompute every \selectop result in
those regions. Since the precomputed results do not fit in regular sample space,
they are overflowed to a separate list, and the sample stores a pointer to the
overflow list instead.

\rev{Unfortunately, overflowing the select index disrupts locality because two
  adjacent samples\mymarginpar{R2 D16} are no longer guaranteed to form a
  bounding interval, as either of them can be overflowed. Therefore,} after
querying the sample, we must either use linear search or incur an additional
cache miss to read the spill list.

We resolve this issue by storing \emph{block indexes} instead of exact bit
positions in regular samples, as this frees up additional bits to store the size
of the bounding interval \rev{(in blocks)}. The \rankop indexes in each block
provide enough information to restore any lost precision at no cost of
additional cache misses. \new{For example, in \figref{functional-index}, block
  1's functional index for \texttt{Child} (i.e. \texttt{Child(5)}) points to
  block 1, which is equivalent to bit position 5 (instead of its original value
  9 in~\figref{functional-index}). Using the inlined rank index, we immediately
  restore \texttt{HasChild.\rankopnospace$_1$(5) = 5}. Since \texttt{Child(5) =
    Louds.\selectopnospace$_1$(7)}, we can reach position \texttt{Child(5)}
  (i.e., 9) by selecting two more ones.}

As shown in \figref{overflow}, each sample occupies 32 bits, where bit 31
indicates whether the sample overflows. If bit 31 is false, then bits 7 to
30 (the 24-bit \head, which means block size must be at least 256) point to the
lower bounding block, whereas bits 0 to 6 (the 7-bit \dist) \rev{store the
  length of the bounding interval (in blocks).}  If \dist does not fit in 7 bits
(\textit{i.e.}, the bounding interval is at least 128 blocks long), then bit 31
is set to true, and the sample is overflowed. The rest of the sample (the 31-bit
\spillptr) \rev{indexes} into the centralized overflow list \spilllist, where
\texttt{\spilllistnospace[\spillptr + i]} is the position of the $i$-th target pattern
in the overflowing interval. \rev{We apply this optimization to} both traditional \selectop indexes and functional indexes with intermediate selects.

We sketch why \spillptr always fits in 31 bits. With block size $B=256$ and
sample rate $S=256$, each overflowing interval has target pattern density at
most $S/128B = 1/128$ (for \selectop indexes) or $B/128B = 1/128$ (for
functional indexes). Since the \bv size does not exceed $2^{32}$, the size of
\spilllist would not exceed $2^{32}/128 = 2^{25}$ and hence \rev{always fits} in
26 bits.

\mymarginpar{R3 O1}\new{Finally, we address the space overhead incurred by the
  spill list. Because each spilled interval contains at most $B = 256$ bits, the
  density of ones is at least 256/(256$ \times$ 128) = 1/128, and the associated
  spill list occupies at most 32/128 = 25\% of the bitvector size. For
  LOUDS-encoded tries, since only half of the bits are zeros, the bound can be
  reduced to 25\% $\times$ 128/127/2 = 12.6\%. For LOUDS-Sparse encoded tries,
  the functional index is sampled on only one of the two bitvectors, which also
  tightens the bound to 12.5\%. Both upper bounds are reached only under extreme
  conditions where most trie edges are concentrated in a few large nodes, which
  is unlikely in practical datasets.}

\subsection{Applying the Rules}

\begin{figure}[t]
    \centering
    \includegraphics[width=\columnwidth]{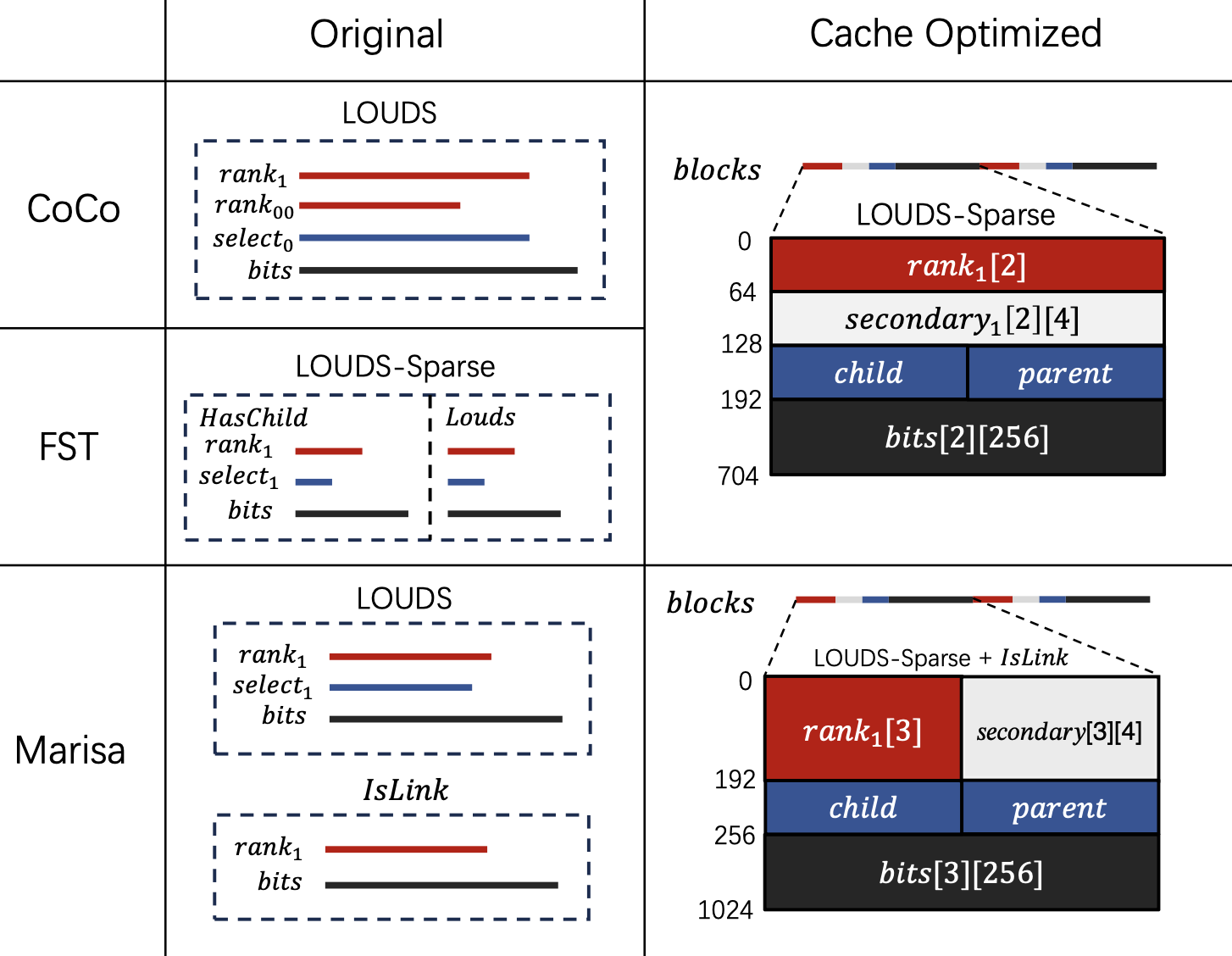}
    \caption{Optimized \bv layout. \rev{\child: functional index for
        \var{Child(x)}. \parent: functional index for
        \var{Parent(x)}. \texttt{secondary}$_1$: secondary index for
        \rankopnospace$_1$.  }}
    \label{fig:bv-design}
  \end{figure}

We apply the above optimizations to FST, CoCo-trie, and Marisa;
\figref{bv-design} shows the resulting \bv layouts.

The FST is encoded using LOUDS-Sparse and only needs to support existence
query. We pack the blocks of \haschild and \louds (\secref{interleaving}) and
interleave bits with the \rankopnospace$_1$ (\secref{interleaving}) and \child
(\secreftwo{functional-index}{overflow}) indexes.

\rev{The CoCo-trie is encoded using standard LOUDS and supports top-down lookup,
  and therefore an ideal approach is to inline (primary and secondary)
  \rankopnospace$_1$ and \rankopnospace$_{00}$ indexes (\secref{interleaving})
  and the \child (\secref{functional-index}, \secref{overflow}) functional
  index. However, this adds a 192-bit space overhead to every 256-bit block,
  which is too expensive. Therefore, we switch to LOUDS-Sparse encoding which
  incurs less index overhead.}

The case with Marisa is \rev{more}\mymarginpar{R2 D15} involved. First, we
change the encoding scheme from LOUDS to LOUDS-Sparse, as this allows for better
alignment. Then, we pack the blocks of \haschild, \louds and \islink, as all
three \bvs are now edge-aligned and are typically accessed with high locality
(\secref{interleaving}). Finally, we interleave the bits with indexes for \child
and \parent(\secref{functional-index}, \secref{overflow}), and each \bv's
\proc{rank}$_1$ (\secref{interleaving}) indexes, as the Marisa trie needs to
support both \rev{forward and reverse lookups}. 

\new{\para{Theoretical analysis} The $C_1$ optimizations preserve asymptotic
  bounds on space and query complexity. The rank and select indexes are
  lower-order terms relative to the bitvector size~\cite{surf}. \mymarginpar{R1 W2, R1 D4}}

\new{\begin{lemma} Given a \bv of size $S$ and a block size of $B$, the $C_1$ optimizations incur at most $S/2B$ additional space overhead.
\end{lemma}

\begin{proof}
  The size of the rank index in $C_1$ is unchanged and takes 32 bits per
  block. 
  Since only half of the bits in a trie encoding are 1, the \selectopnospace$_1$
  index contains about $S/2B$ elements, whereas the \texttt{Child} functional
  index contains $S/B$ elements.
\end{proof}
} \rev{Although functional indexes use more space compared to traditional
  \selectop indexes,\mymarginpar{R2 D16} the overall space overhead is small
  relative to the index size and label data. For example, when $B=256$, the
  functional index takes 12.5\% extra space relative to the \bv size compared to
  the \selectop index.}


\begin{lemma}
  A child navigation takes at most 4 cache misses in all of the
  \bvdesign-optimized tries in this paper.
\end{lemma}

\begin{proof}
  Let \bvdesign-CoCo, \bvdesign-FST, and \bvdesign-Marisa denote the CoCo-trie,
  FST, and Marisa trie with the $C_1$ optimization as illustrated
  in~\figref{bv-design}.  As shown in~\figref{bv-design}, \bvdesign-CoCo,
  \bvdesign-FST, and \bvdesign-Marisa have block sizes of 704, 704, and 1024
  bits, respectively. Given a cache-line size of 64 bytes (512 bits), each block
  in the \bvdesign-optimized succinct tries takes at most 2 cache lines.  A
  child navigation with the functional index queries two blocks: the input block
  and the output
  block. 
\end{proof}

Although the blocks in the \bvdesign-optimized tries span 2 cache lines each,
the second cache line in the block is prefetched and therefore incurs no
additional cache misses. In contrast, the original succinct tries (e.g., CoCo,
FST, and Marisa), are encoded using separate \bvs, which take at least 3 random
cache misses per block.


%% file: unary-path.tex
\section{The Second ``C'': Adaptive Unary-Path Compression}\label{sec:unary-path}

In this section, we address the space inefficiency incurred by unary paths
(especially suffix unary paths) and introduce the second "C" of \csquared: the
adaptive unary-path \underline{C}ompression scheme (denoted with
\compressionscheme). \rev{The \compressionscheme optimization enables all tries
  to access unary-path compression in the data. Specifically, it combines
  Marisa's recursion~\cite{marisa} with the Fast Static Symbol Table
  (\fsst)~\cite{fsst} as the tail container. Furthermore, it \emph{adaptively}
  chooses the number of levels of recursion for the Marisa trie based on
  per-dataset space savings.

  We evaluate \csquared tries against baselines at the same recursion level;
  since CoCo-trie and FST originally lack recursion, \compressionscheme
  changes only their tail container.}

\para{Locality issues in the data} \rev{Long unary paths increase topology size
  and waste space through redundant label sequences.} \newrev{Without proper
  compression, suffix paths can account for about 80\% of succinct-trie space,
  as shown in~\tabref{unary-wiki-log}.}

\newrev{Prior work uses two orthogonal compression methods (detailed
  in~\secref{preliminaries}): PDT~\cite{pdt} uses \repair~\cite{larsson2000off},
  while Marisa~\cite{marisa} uses recursion, as detailed
  in~\secref{preliminaries}. Neither method dominates: on the
  \texttt{log}~\cite{log-dataset} dataset, recursion and \repair achieve
  compression ratios of 16.4$\times$ and 7.5$\times$, respectively; on
  \texttt{xml}~\cite{text-collection}, the ratios are 3.9$\times$ and
  5.0$\times$.} 

\para{Choice of tail container} Although Marisa originally uses a sorted
  tail container, recursion is agnostic to the tail container; more advanced containers such as approximate
  \repair~\cite{pdt} and \fsst~\cite{fsst} can improve both space and query
  time.

  We limit our \csquared tries to the \fsst tail container because it captures almost all (within \fsstvsrepair) of the compression benefits of
  approximate \repair with $12.8\times$--$21.6\times$ faster build times. 
  As
  shown in~\tabref{c2-full} in~\secref{experiments}, choosing \fsst as the tail
  container always achieves a better compression ratio (up to \fsstmaxspacesavings) with minimal impact on the query time (.96-1.13$\times$)
  over the sorted container.



\begin{figure}[t]
    \begin{minipage}[h]{.23\textwidth}
      \includegraphics[width=\textwidth,page=9,trim=14cm 7.6cm 15cm
      7.5cm,clip]{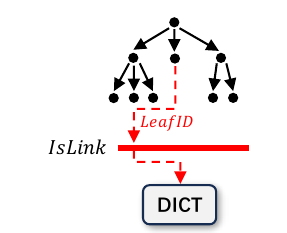}
        \caption{Suffix path compression for \csquared-FST and \csquared-CoCo.}
        \label{fig:adaptive-compression-fst-example}
    \end{minipage}
    \hspace{0.1cm}
    \begin{minipage}[h]{.23\textwidth}
      \includegraphics[width=\textwidth,page=10,trim=14cm 7.8cm 14cm
      7.2cm,clip]{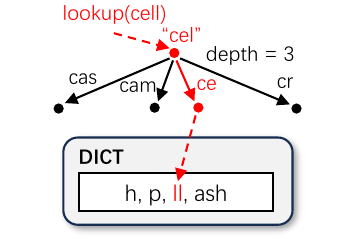}
        \caption{A \csquared-CoCo suffix path compression example. Stored keys:
          cash, camp, cell, crash.} 
        \label{fig:adaptive-compression-coco-example}
    \end{minipage}
\end{figure}

\para{Integration with succinct tries} \rev{First, we enable \bvdesign-FST,
  \bvdesign-CoCo, and \bvdesign-Marisa to access recursion and compressed tail
  containers. We focus the discussion on 
  FST and CoCo, as these previously could not access unary-path compression.}

\figref{adaptive-compression-fst-example} illustrates how we containerize the
storage of suffix paths in FST. Specifically, each leaf node has an \islink bit indicating whether its
suffix path is moved to the next container. A suffix link at position $i$ can thus be uniquely identified by
\var{LinkId(i) = IsLink.rank$_1$(LeafID(i))}. When a lookup operation reaches a
leaf node, we follow the \var{IsLink} bitvector to check the remaining suffixes
in the container (if any).

The integration with \rev{\bvdesign-CoCo} is \rev{similar}\mymarginpar{R2 D15}
to \rev{\bvdesign-FST}, but \rev{CoCo requires special care for correctness: its
  encoding schemes support only exact
  lookup}.~\figref{adaptive-compression-coco-example} shows the issue. When
looking up the string ``cell,'' the CoCo-trie would search for ``cel'' in the
root \rev{because the root node has depth 3}. Since ``cel'' is not a key of the
root, the lookup would fail, even if ``cell'' exists in the trie.
Instead, we search for the lower bound of the target key rather than an exact match. If the lower bound is a prefix of the target key, we trace the link to
check the remaining suffixes.

\rev{\para{Adaptive recursion depth}} We enable all of the \rev{succinct} tries
studied in this paper to access unary-path compression, but focus on the case of
the Marisa trie and show \new{how to achieve better space-time tradeoffs by
  combining recursion with \fsst.}

Using recursion in the other tries (FST and CoCo) trades query performance for
space savings. Therefore, we limit non-Marisa tries to no recursion in our evaluation,
but expose it as a user option for different space-time tradeoffs. We will include the
data on the effect of recursion in non-Marisa tries in the full version.

\new{The original Marisa-trie exposes the maximum number of recursions as a
  parameter to the user, but offers no per-dataset guidance on the space-query
  tradeoff. As the original Marisa trie documentation notes~\cite{marisa-code}
  and \secref{experiments} confirms, recursion has a serious negative impact on
  query performance, so we only want to continue recursion when the space
  reduction is significant.

  We \emph{adaptively select} the number of recursion levels that achieve
  the best time-space tradeoff per dataset. To balance
  compression ratio against build/query time, we stop recursion when the space
  savings would be less than some $\epsilon = 0.1$\mymarginpar{R1 D3, R2 D5}
  (relative to trie size). However, the user can set $\epsilon$ to suit
  their use case.}

\mymarginpar{R1 W3, R1 D3, R2 D1, R2 D4} To efficiently determine the space
usage at each level of recursion, we adopt \fsst's fast estimation scheme,
encoding a subset of the data to approximate the compression ratio.  In
practice, this estimate is within 10\% of the true compression ratio.

\para{Other locality optimizations} \new{For the Marisa trie, compressed tail
  containers degrade lookup performance in nodes with many links. To address
  this issue, we store the branching label (i.e., the first label of each unary
  path) in the label vector, enabling in-place intra-node search accelerated by
  SIMD~\cite{pdt}.}

\mymarginpar{R1 D3, R2 D1, R2 D3, R3 O4}\rev{To further improve locality, we
  store unary paths smaller than a link in-place, avoiding the index overhead of
  moving them to the next level. The number of bits needed to store a link
  depends on the size of the trie (i.e., if a trie has $n$ nodes, we need
  $\lg(n)$ bits to store each link). Short unary paths may be further
  compressed by moving them to the next level, but the space savings is limited
  by the short path length relative to the link cost.}


%% file: experiments.tex
\section{Experimental Evaluation}\label{sec:experiments}

We evaluate three succinct tries optimized using \csquared (\csquared-FST,
\csquared-CoCo, and \csquared-Marisa) and compare them with \numtries state-of-the-art
tries: FST \cite{surf}, Marisa \cite{marisa}, CoCo-trie \cite{coco_0, coco_1},
PDT \cite{pdt}, ART \cite{art}, C-ART \cite{hybrid_index}, \newmymarginpar{Meta
  4, R2 O2}\newrev{c-trie++~\cite{tsuruta2022c}, and z-fast
  trie~\cite{belazzougui2010dynamic}}. The first four are \bv-based succinct
tries introduced in \secref{preliminaries}; ART is a well-established
pointer-based dynamic trie, C-ART is the compact and static version of ART, and
\newrev{c-trie++/z-fast trie are high-performance dynamic tries}. We compare all
data structures on query performance, space usage, and build time.

\para{Summary} The \csquared tries significantly outperform the original
versions in both performance and space efficiency on most datasets.  On average,
\csquared-FST, \csquared-CoCo, and \csquared-Marisa improve query
performance/space usage by \csqfstspeedup/\csqfstspacesaving,
\csqcocospeedup/\csqcocospacesaving, and \\
\csqmarisaspeedup/\csqmarisaspacesaving, respectively (at equal recursion).

These improvements stem from improved locality and fewer cache misses in both
the index and the data. Applying \bvdesign (the cache-conscious bitvector
redesign in~\secref{bv-design}) reduces cache misses per query by $1.51\times$
for FST and $1.26\times$ for Marisa on large datasets. The \bv redesign adds at
most $4\%$ space overhead. At equal recursion, \compressionscheme (the
unary-path compression scheme in~\secref{unary-path}) improves query
performance, space usage, or both over \bvdesign alone.

\subsection{Experimental Settings}\label{sec:experimental-setup}

\input{latexfigs/table-datasets}

\para{Hardware} \newrev{We ran the experiments on a server with an AMD EPYC 9555
  64-Core Processor running at 3.2GHz. The server has 1007 GB of memory, a 3 MiB
  L1 data cache, a 2 MiB L1 instruction cache, a 64 MiB L2 cache, and a 256 MiB
  L3 cache.} We ran the experiments with one thread.

\para{Implementation} We implement the three \csquared tries (\csquared-FST,
\csquared-CoCo and \csquared-Marisa) using C++17. We use the compressed string
pool implementation of PDT~\cite{pdt-code} \new{and the publicly-available
  implementation~\cite{fsst-code} of FSST~\cite{fsst} for the \repair and \fsst
  tail containers, respectively.} We also \rev{use} the sdsl~\cite{sds_overview,
  sdsl_lib}, ds2i~\cite{ds2i} and sux~\cite{rank_select_0} libraries for
\csquared-CoCo as in the original CoCo-trie~\cite{coco-code}.

For CoCo-trie~\cite{coco-code}, Marisa~\cite{marisa-code}, PDT~\cite{pdt-code}
and C-ART~\cite{cart-code}, we use their original implementations. We adopt the
ART implementation in the C-ART codebase. For FST, we use an optimized
third-party implementation~\cite{fst-code} which compacts suffixes into a
contiguous array, as this allows for fairer comparison with many unary suffix
paths. \newmymarginpar{Meta 4, R2 O2}\newrev{Finally, we also evaluate
  c-trie++~\cite{tsuruta2022c} and z-fast trie~\cite{belazzougui2010dynamic},
  two state-of-the-art dynamic tries. }

\newmymarginpar{Meta 3, R3 O2}\newnew{All tested systems support
  existence queries, but out of the tested \csquared tries, only FST supports an
  iterator for range queries. Therefore, we evaluate FST versus \csquared-FST on
  range queries. In each experiment, we generate range queries with a start key
  and a length, or number of subsequent keys to return beginning from the start
  key. We vary the query lengths \texttt{k} between experiments and test
  $\texttt{k} = 1, 10, \ldots, 10000$.}

The ART and C-ART implementations we use are designed as filters but not indexes
and may return false positives. To support lossless queries, the original
dataset must be replicated and may need to be checked when the filter returns
positive results. We report space costs without the replicated dataset; query
times include false-positive checks.

Finally, as detailed in~\secref{coco-prime}, we introduce CoCo', a version of
the CoCo-trie with an optimized build routine but the same \bv design, because
the open-source CoCo-trie does not build on large datasets.

All code is compiled with g++ 10.3.0 at -O3.

All implementations used in our experiments are publicly available at
\url{https://github.com/alexztc/C2}.

\para{Parameter settings} For FST, we use the default setting with $R =
64$. \rev{For CoCo-trie, CoCo', and \csquared-CoCo, we set $\alpha = 5\%$ for a
  good space-time balance.}  We do not enable LOUDS-Dense for \csquared-FST as
it provides no benefit. \new{We use Marisa's default cache size for both Marisa
  and \csquared-Marisa, which is 1/512 of the key count. We set $\epsilon = .1$
  for all \csquared tries.}

\para{Datasets} \tabref{datasets} details the \numdatasets datasets in our
evaluation.  These datasets span diverse sources, key lengths, prefix depths,
and alphabet sizes.

\rev{Furthermore, following the methodology from the CoCo-trie
  paper~\cite{coco_1}, we generate prefix-only versions of each dataset (denoted
  by \texttt{dataset*}) because CoCo originally used only prefix-only
  datasets. We compare the open-source CoCo implementation with our CoCo'
  implementation on the prefix-only datasets.}

\para{Experimental setup} We set up experiments as follows: First, each dataset
is sorted and deduplicated. Second, we build the trie under test on the dataset
and record the build time and memory usage. Finally, we query the trie using all
keys in the dataset in random order and compute the average query latency. We
omit negative-query latencies, as they are generally proportional to positive
workloads~\cite{coco_1}. We also measure the cache performance of queries in
different tries on the \logdata and \wikidata datasets by recording the number
of cache misses using \texttt{perf}. \newrev{All times are the average of three
  trials after one warm-up trial.}

\para{Notation} We use Marisa-$i$ to denote Marisa with $i$ levels of
recursion. Furthermore, for a given trie X, we use \bvdesign-X to denote X
optimized with \bvdesign (\secref{bv-design}). We use \csquared-X to denote X
with both \bvdesign and \compressionscheme. For example, \bvdesign-FST is
\bvdesign-optimized FST with sorted tail container, whereas \csquared-Marisa-1
is \bvdesign-optimized Marisa with \fsst tail container and one level of
recursion.

\subsection{Building \csquared-CoCo on larger datasets}\label{sec:coco-prime}

\input{latexfigs/table-c1-ablation}
Since the original CoCo-trie implementation
runs out of memory when building the larger prefix-only datasets, we implement
an optimized build routine underneath \csquared-CoCo to support all datasets.
The original CoCo is built from a pointer-based uncompacted trie, which not only
consumes excessive memory but also incurs significant cache misses when
traversing each node's descendants at different levels, a key operation of
CoCo-trie's optimization process. We optimize the builder by representing the
uncompacted trie as \csquared-FST, which significantly reduces cache misses
because in LOUDS-Sparse representation each node's descendants at the same level
fall into contiguous memory. Our implementation reduces the build time of
\csquared-CoCo by up to $14\times$ and uses orders of magnitude less memory
during build than the original CoCo-trie.

We introduce \defn{CoCo'}: \csquared-CoCo's build routine with the original
CoCo-trie \bv, to isolate \bv differences. 

\mymarginpar{R2 D8, R2 D11}\rev{We use CoCo' as the CoCo-trie baseline for
  larger datasets. As shown in~\tabref{coco-vs-coco-prime}, on the prefix-only
  datasets that the original CoCo-trie can build on, CoCo-trie and CoCo' have
  similar query latencies (within \newrev{100} ns). However, the original CoCo-trie suffers from poor
  performance and space usage on full datasets because it does not handle unary
  suffix paths. For example, on \urldata, a relatively small dataset, the
  original CoCo-trie implementation requires more than 10 minutes and 60 GB of
  memory to build. Furthermore, it takes over 2000 ns/query and occupies nearly
  half as much memory as the original dataset, both the worst among all tested
  tries.}

\input{latexfigs/table-alation}

\input{latexfigs/table-c2-full}

\subsection{Ablation Study and Sensitivity Analysis}\label{sec:ablation-study}

\figref{intro-results-summary} shows ablation results of \bvdesign and \csquared
across configurations of FST, CoCo, and Marisa. \tabref{c2-full} contains the
full data for the different \csquared configurations discussed
in~\secreftwo{ablation-study}{performance-comparison}.

\newmymarginpar{Meta 5, R3 O4}\newnew{\para{Effect of hybrid
    \bv}~\tabref{fst-ablation-rel} evaluates the effect of switching to
  LOUDS-Sparse in FST~\cite{surf}. As mentioned in~\secref{preliminaries}, the
  original FST introduced a hybrid \bv with LOUDS-Sparse and LOUDS-Dense. We
  find that the hybrid \bv improves the performance of both build and query in
  the baseline FST (with no special tail container). However, the effect is much
  smaller on \csquared-FST with the FSST tail container, so we switch to the fully
  sparse version in the final \csquared configuration.}

\para{Effect of cache-optimized \bv design}
  First, we isolate the effect of \bvdesign, the cache-optimized \bv design.
  We perform the ablation study on the datasets in~\tabref{datasets} (1) with the
  sorted tail container and no recursion to avoid confounding factors from path
  compression, and (2) with all branching labels stored out of place to suppress
  intra-node search effects\newrev{, so that the remaining query-time difference
  is attributable to the \bv layout alone}.

  The ablation study in \tabref{c2-full} shows that the \bvdesign optimization
  improves query performance by \bvfstspeedup, \bvcocospeedup, and
  \bvmarisaspeedup, respectively, compared to the original FST, CoCo-trie, and
  Marisa. \newrev{These trie-level gains compound directly from the per-operator
  speedups in \tabref{trie-op}, since every full lookup issues
  $O(\mathit{depth})$ rank queries on the topology bitvector and a small
  constant number of select queries, each of which benefits from the co-located
  block metadata in our layout.}

  The bitvector redesign has the least impact in the CoCo-trie because the query
  time in CoCo is dominated by the cost to read the encodings, and each
  CoCo macro-node already amortises several traversal steps over one \bv
  access. However, FST and Marisa improve significantly because \bv cache
  misses dominate query time: each rank query in their baselines incurs
  two cache misses --- one for the block summary and one for the bit field ---
  which our co-located layout collapses to a single line.

  \input{latexfigs/table-trie-op}

  \tabref{trie-op} reports the effect of \bvdesign on the fine-grained \bv
  operations that make up succinct-trie navigation. The results demonstrate that
  \bvdesign improves the performance of key \newrev{rank-based operations such
  as \texttt{leaf\_id} and \texttt{internal\_id} by $1.2-7.1\times$.}
  We observe similar trends for trie-specific operations, such as parent
  navigation in Marisa\newrev{ ($1.76\times$ on \texttt{parent\_pos}) and child
  navigation across all three tries ($1.3-2.8\times$ on \texttt{child\_pos})}.
  \newrev{The only operators where the baseline is faster are \texttt{get} and
  \texttt{degree}, which perform linear scans through bits and do not use rank
  and select\newrev{; interleaving block metadata with the bit field slightly
  reduces the effective bit density per cache line on such scans}. However,
  neither of these lie on the critical path of a full lookup or range
  query\newrev{, and both are invoked $O(1)$ times per query at sub-3-ns cost
  each, so the few nanoseconds lost are dwarfed by the tens of nanoseconds
  saved on every rank query}}.

\begin{figure*}[t]
    \centering
    \includegraphics[width=\textwidth]{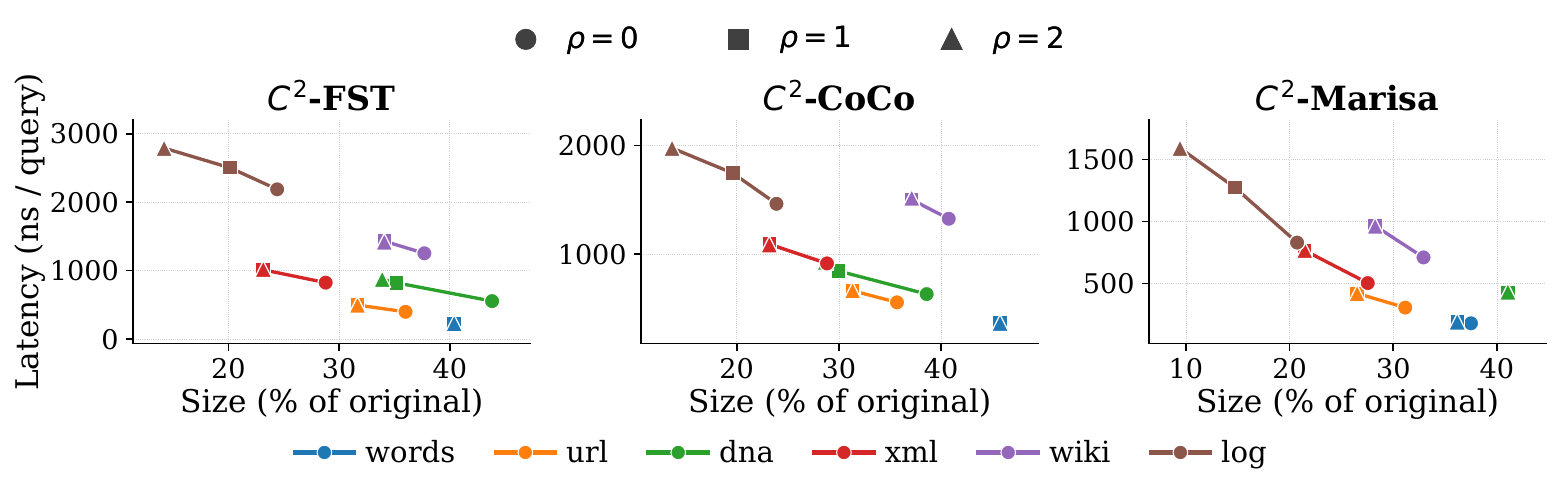}
    \caption{
    Space-latency Pareto frontier under recursion depths
    $\rho\in\{0,1,2\}$ for the three trie variants across all six datasets.}
    \label{fig:recursion-pareto}
\end{figure*}

\para{Effect of unary-path compression} \new{Next, we
  evaluate the effect of the \compressionscheme compression
  scheme. Specifically, for all tries, we measure the effect of replacing the
  sorted tail container with \fsst. Furthermore, we test different recursion
  levels to evaluate the space-time tradeoff exposed by more recursions.}

\newrev{Replacing the sorted tail container with \fsst improves space usage by
  \csqfstspacesaving, \csqcocospacesaving, and \csqmarisaspacesaving over
  \bvdesign-FST, \bvdesign-CoCo, and \bvdesign-Marisa, respectively. The effect
  on query performance is minimal (1-1.05$\times$ on average).
}



As shown in~\figref{recursion-pareto}, recursion levels 0 and 1 expose Pareto-optimal space-time trade-offs.
  For example, on \csquared-Marisa, on the \logdata dataset, no recursion takes
  829 ns/query, while one level of recursion takes 1308 ns/query. On the other
  hand, the space usage improves from 20.7\% to 14.7\%. Similarly, on the
  \wikidata dataset, no recursion takes 732 ns/query compared to 979 ns/query
  with one level of recursion. For \wikidata, one level of recursion improves
  the space usage from 32.9\% to 28.3\%.

  The same pattern holds for \csquared-FST and \csquared-CoCo. On the \logdata
  dataset, \csquared-FST without recursion takes 2187 ns/query at 24.4\% space,
  while one level of recursion takes 2502 ns/query at 20.1\% space; on the same
  dataset, \csquared-CoCo without recursion takes 1463 ns/query at 23.9\% space,
  while one level of recursion takes 1745 ns/query at 19.6\% space. The same
  trade also holds on \wikidata, \xmldata, \urldata, and \dnadata: across our
  six datasets, the absolute space savings from recursion levels 0 to 1 range
  from 0 percentage points (e.g., \wordsdata, whose avg.\ longest common prefix
  of 6 leaves no recursive sub-prefix structure to extract) to roughly 6
  percentage points (\logdata across all three variants), while the
  corresponding latency degradation stays below 1.5$\times$ in every cell.

  Beyond one level of recursion, further recursion's space savings rarely justify the
  performance penalty. For example, if we continue to recurse \csquared-Marisa-1
  on \logdata, the relative space usage improves by 1.6$\times$ (which
  corresponds to a 5.3 percentage point improvement relative to the original
  data size, from 14.7\% to 9.4\%), but the query latency degrades by 1.25$\times$
  (from 1308 to 1593 ns/query); on \xmldata, the same step yields no measurable
  further space improvement. In fact, \logdata is the only dataset on which
  $\rho{=}2$ strictly Pareto-dominates $\rho{=}1$ across all three tries, owing
  to its 113-character average key length harbouring nested common prefixes
  that a second recursion round can still extract. We observe one further
  anomaly: on \dnadata, \csquared-Marisa is fully insensitive to recursion ---
  its size remains at 41.1\% of the original data and its latency varies by less
  than 1\% across $\rho \in \{0,1,2\}$ --- because its unary-path encoder
  already absorbs DNA's unique-suffix structure at $\rho{=}0$, leaving no
  inter-tail prefix redundancy for the adaptive cost-based recursive
  string-pool builder to exploit.

  Going forward, we limit the recursion level of \csquared-FST and \csquared-CoCo
  to 0 to match the original designs. For \csquared-Marisa, we use the recursion
  level prescribed by the adaptive compression scheme. We expose recursion as
  an option for space-constrained settings: dialling \csquared-Marisa up to
  $\rho{=}2$ on long-key datasets like \logdata trades a further 5 percentage
  points of space for a 1.25$\times$ query latency penalty, while on every
  other dataset the same dial yields at most one percentage point of additional
  space and is therefore Pareto-dominated by $\rho{=}1$.

\begin{figure*}[t]
    \centering

    \begin{subfigure}[t]{0.32\textwidth}
        \centering
        \includegraphics[width=\linewidth]{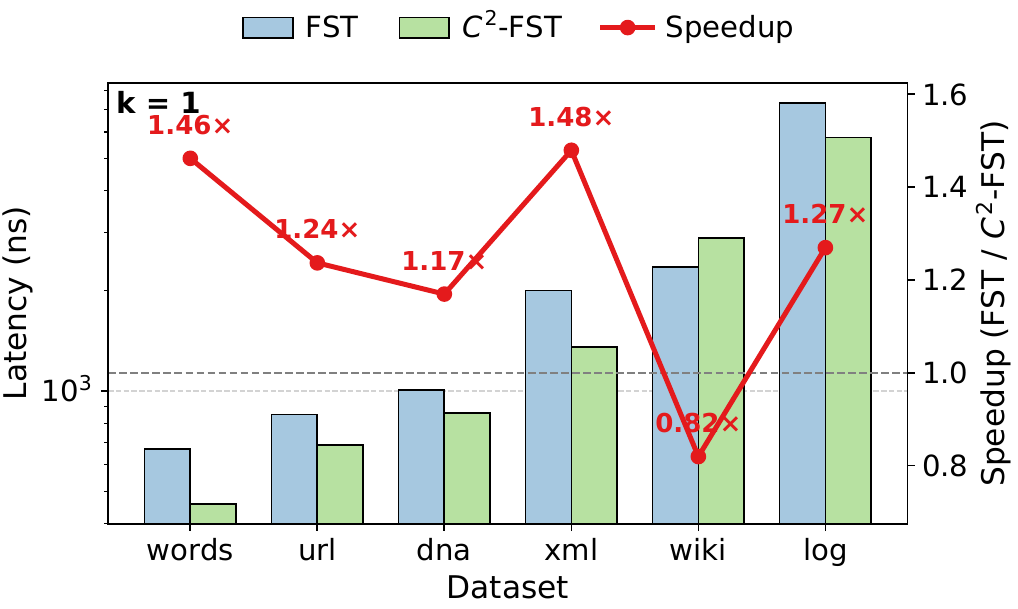}
        \label{fig:range-fst-w10}
    \end{subfigure}
    \hfill
    \begin{subfigure}[t]{0.32\textwidth}
        \centering
        \includegraphics[width=\linewidth]{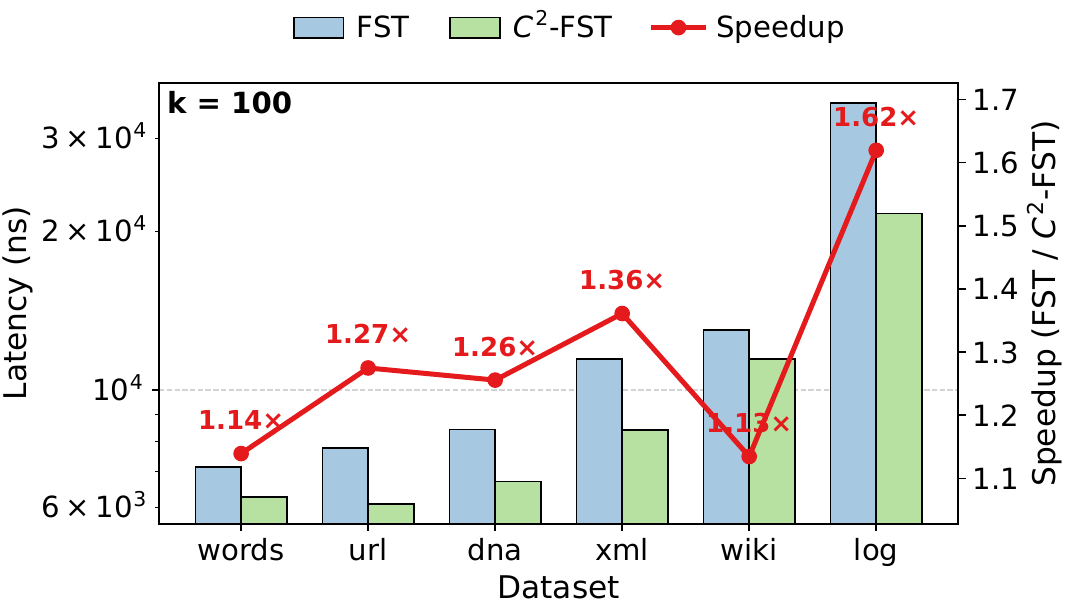}
        \label{fig:range-fst-w100}
    \end{subfigure}
    \hfill
    \begin{subfigure}[t]{0.32\textwidth}
      \centering
      \includegraphics[width=\linewidth]{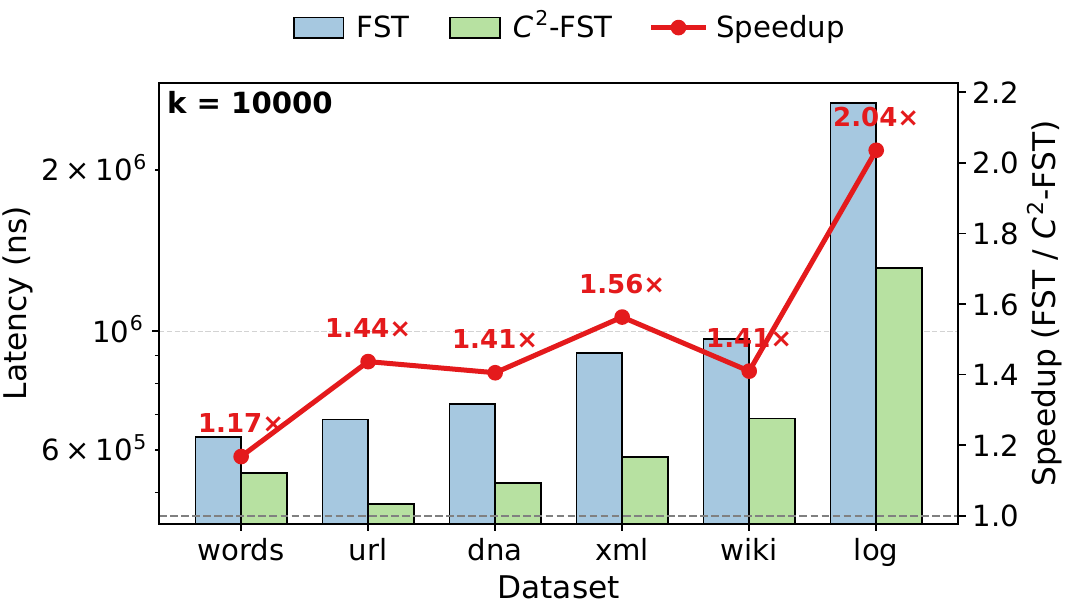}
        \label{fig:range-fst-w1000}
    \end{subfigure}
    \caption{ \newnew{ Range-query latency comparison between FST-baseline and
        \csquared-FST across six datasets under different range widths $k$.  Each
        subplot corresponds to a fixed range width.  Bars report absolute
        latency (left axis), while the red line indicates the speedup of
        \csquared-FST over the baseline (right axis).} }

    \label{fig:range-fst-vs-c2fst}
\end{figure*}

\subsection{Comparison to State-of-the-art
  Tries}\label{sec:performance-comparison}\mymarginpar{R2 D10}

We compare \csquared tries with all baselines.\rev{~\tabref{c2-full} reports the
  full data and the configurations prescribed by the \compressionscheme, the
  adaptive compression algorithm.}

\para{Existence queries} \rev{Among all tries, ART always has the best query
  performance at the cost of high space usage. Child navigation in ART simply
  requires following a pointer, while succinct tries must perform complex \bv
  operations. \mymarginpar{R1 D5, R3 D2}\new{Furthermore,
    C-ART~\cite{hybrid_index} uses an internal-node organization that trades
    search performance for space savings over ART.}} \newmymarginpar{Meta 4, R2
  O2}\newrev{Both c-trie++ and z-fast trie are slower than ART for queries while
  consuming significantly more space.}

\new{\tabref{c2-full} demonstrates that \compressionscheme does not affect the
  query-time gains from \bvdesign across tries and datasets.}  \rev{Among the
  succinct tries}, \csquared-Marisa achieves the lowest query latency with a
space usage within $1.2\times$ of the best. On average, \csquared-Marisa
improves the query latency of the original Marisa by \csqmarisaspeedup.

On \texttt{words} and \texttt{url}, \csquared-Marisa-1's space consumption is
relatively high compared to \rev{Marisa}-1 \mymarginpar{R3 O4}\rev{because in
  this case, the \compressionscheme scheme stores short unary paths in place
  rather than in the tail container (see~\secref{unary-path}). For example, on
  the \texttt{words} dataset, \csquared-Marisa contains only about .3$\times$ as many
  links as Marisa.}

\csquared-FST improves query latency/space usage by
\csqfstspeedup/\csqfstspacesaving on average compared with FST.
\csquared-CoCo does not benefit as \rev{much} from \bvdesign as it incurs fewer
\bv operations, but still improves both space usage and query latency compared
to \rev{CoCo'}.

\para{Time-space tradeoff} \new{We next examine query time and space
  usage. \compressionscheme improves space usage by $1.34\times$ \rev{on
    average} compared to their original versions.}

\new{On average, with no recursion, switching from the
  sorted tail container to FSST reduces the space usage by \toplevelspacesavings at the
  cost of a $1.05\times$ query performance penalty and $1.1\times$ higher build
  time across all tested succinct tries.

  Additionally, for Marisa, on the datasets where \compressionscheme prescribes
  one level of recursion, the recursion improves the space efficiency by
  $1.29\times$ on average, but also incurs a query performance penalty of
  $1.43\times$ compared to no recursions.}  \rev{Specifically, recursion turns
  out to be extremely effective on the highly redundant \logdata dataset. For
  \csquared-Marisa, one level of recursion reduces the space cost by
  $1.74\times$, but increases the query time by $1.7\times$ compared to no
  recursions due to the presence of many recursive paths.}

\new{On \wordsdata, \urldata and \wikidata, \csquared-Marisa-1 has higher space
  usage (by at most $1.19\times$) than Marisa-1 because of the space overhead of
  storing the branching labels in place
  (see~\secref{unary-path}). 
}

On the \rev{larger datasets (e.g., \logdata)}, \csquared-FST and \csquared-CoCo
are generally less competitive than \csquared-Marisa in both query performance
and compression ratio, due to their handling of internal unary
paths. \csquared-FST does not handle such paths at all, \rev{and} \csquared-CoCo
does not address the data redundancy in such paths. Furthermore, unary paths
that do not fit the machine word size will be fragmented in \csquared-CoCo,
impairing data locality.

\newrev{\para{Range queries} \newmymarginpar{Meta 3, R1 O3}
  ~\figref{range-fst-vs-c2fst} shows that the \csquared optimizations speed up
  range queries by $1.2-1.5\times$ in FST, which was the only tested trie that
  natively supports successor queries. The locality improvements in
  \csquared-FST carry over into range queries, where each \texttt{next()} step
  can be resolved from data in the packed 256-bit blocks, which are already in
  cache. In contrast, the original FST fetches the equivalent information from
  three separately heap-allocated arrays, incurring up to three cache-line
  misses per step.

  The relative advantage of \csquared-FST increases monotonically with range
  width, with speedups on all but one configuration (\wikidata with $k =
  1$). The \wikidata contains many diverse suffixes, so the FSST tail container
  in \csquared-FST finds almost nothing to compress, as shown
  in~\tabref{c2-full}. However, \csquared-FST must incur additional indirections
  during traversal to check the \texttt{is\_link} array and fetch the compressed
  suffix, while the original FST can do a direct comparison.  }

\para{Build time}\label{sec:build-performance} While modern static succinct
tries excel at space efficiency and query latency, their build \rev{times are}
significantly worse (\rev{between 3.2-24.7$\times$}) than \rev{the build times}
of pointer-based tries such as ART and C-ART, especially \rev{on large
  datasets}.\newmymarginpar{R3 D8} \newrev{Depending on the dataset, the
  \csquared optimizations can increase the build time by about $2\times$ due to
  the time it takes to compress the tail container. We chose FSST as the tail
  container to try to minimize build times compared to approximate \repair,
  which can take another $2\times$ longer, as shown in PDT's build times.}
\rev{Furthermore, building is a one-time cost for succinct tries and can be
  amortized over many queries.}
  Future work will parallelize the build phase of succinct tries.


%% file: latexfigs/table-datasets.tex
\begin{table}[t]
  \caption{Datasets used in evaluation. All dataset sizes are expressed in
    MB. LCP means length of longest common prefix.  \textit{Size*} is the size
    of the prefix-only version of each dataset.}
  \resizebox{\columnwidth}{!}{
 \newrev{  \begin{tabular}{l r r r r r l}
             \toprule

    \textit{Name} & \textit{Size} & \textit{Size*} & \textit{\#Keys}
    &\makecell{\textit{Avg} \\ \textit{len}} &\makecell{\textit{Avg} \\ \textit{LCP}} & \textit{Description}\\
    \midrule

    \wordsdata~\cite{words-dataset}
    & 5
    & 4
    & 4.67E5
    & 9
    & 6
    & \small{English words}\\

    \urldata~\cite{web-crawl}
    & 37
    & 17
    & 1.77E6
    & 21
    & 7
    & \small{UK private domains}\\

    \dnadata~\cite{text-collection}
    & 101
    & 44
    & 3.32E6
    & 31
    & 11
    & \small{DNA 31-mers}\\

    \xmldata~\cite{text-collection}
    & 117
    & 74
    & 2.15E6
    & 56
    & 33
    & \small{dblp XML dump} \\

    \logdata~\cite{log-dataset}
    & 586
    & 258
    & 4.45E6
    & 137
    & 54
    & \small{server access logs}\\

    \wikidata~\cite{wiki-dataset}
    & 359
    & 236
    & 1.75E7
    & 21
    & 11
    & \small{Wikipedia titles}\\
    \bottomrule
  \end{tabular}}
}
  \label{tab:datasets}
\end{table}


%% file: latexfigs/table-c1-ablation.tex
  \begin{table}[t]
      \caption{
    Comparison of CoCo and CoCo$'$ (\queryshort, in ns per query) and space usage (\sizeshort,
    in \% of original dataset size) on the
    prefix-only datasets used in the original CoCo-trie
    evaluation.}
        \begin{center}
            \newcolumntype{P}{>{\raggedleft\arraybackslash}p{.24in}}
            \resizebox{\columnwidth}{!}{
             \newrev{  \begin{tabular}{l P P P P P P P P }
                \toprule

                 & \multicolumn{2}{c}{\wordsstar}
                 & \multicolumn{2}{c}{\urlstar}
                 & \multicolumn{2}{c}{\dnastar}
                 & \multicolumn{2}{c}{\xmlstar}\\

                \cmidrule(lr){2-3}
                \cmidrule(lr){4-5}
                \cmidrule(lr){6-7}
                \cmidrule(lr){8-9}

               \textit{Trie} &
               \textit{\queryshort} & \textit{\sizeshort} &
               \textit{\queryshort} & \textit{\sizeshort} &
               \textit{\queryshort} & \textit{\sizeshort} &
               \textit{\queryshort} & \textit{\sizeshort} \\

                \midrule

                  CoCo
                  & 315 & 62.7\%
                  & 742 & 42.4\%
                  & 655 & 24.0\%
                  & 1039 & 47.5\% \\

                  CoCo$'$
                  & 380 & 54.8\%
                  & 811 & 32.3\%
                  & 743 & 40.1\%
                  & 1011 & 29.6\% \\

                \bottomrule
            \end{tabular}}
            }
        \end{center}

    \label{tab:coco-vs-coco-prime}
    \end{table}

%% file: latexfigs/table-alation.tex
\begin{table}[t]
\caption{
Normalized build and query time (relative to \csquared-FST; lower is better)
for different LOUDS-Sparse/Dense configurations of FST.
FST-Sparse uses LOUDS-Sparse/LOUDS-Dense, while FST-Hybrid uses
LOUDS-Sparse, both with a sorted tail container.
\csquared-FST uses the same bitvectors with the FSST tail container.
}
\centering
\setlength{\tabcolsep}{5pt}
\begin{tabular}{lrrrr}
\toprule
& \multicolumn{2}{c}{\textit{words} (0.47M keys)}
& \multicolumn{2}{c}{\textit{log} (4.45M keys)} \\
\cmidrule(lr){2-3}
\cmidrule(lr){4-5}
& \textit{Build}
& \textit{Query}
& \textit{Build}
& \textit{Query} \\
\midrule
FST-Sparse
    & 7.73$\times$
    & 3.50$\times$
    & 2.30$\times$
    & 2.73$\times$ \\

FST-Hybrid
    & 7.02$\times$
    & 2.85$\times$
    & 2.24$\times$
    & 2.57$\times$ \\

\csquared-FST-Sparse
    & 5.31$\times$
    & 1.52$\times$
    & 1.53$\times$
    & 1.82$\times$ \\

\csquared-FST-Hybrid
    & 5.29$\times$
    & 1.45$\times$
    & 1.55$\times$
    & 1.85$\times$ \\
\bottomrule
\end{tabular}
\label{tab:fst-ablation-rel}
\end{table}

%% file: latexfigs/table-c2-full.tex

\begin{table*}[t]
 
\caption{
Build time,
query latency,
and space usage on the datasets in~\tabref{datasets}.
The prefixes
        \bvdesign/\csquared indicate that the
        \bvdesign/\bvdesign+\compressionscheme optimizations are applied.  We
        fix the tail container to be FSST with the \compressionscheme
        optimization, but we report results with the sorted tail container for
        the ablation study on the \bvdesign optimization. Marisa-1 means the
        Marisa trie with 1 recursion level.  For Marisa, we show all data points
        before \compressionscheme stops recursion.  For all other tries, the
        recursion level was set to 0 for a fair
        comparison. For each column, green cells highlight the best value and blue cells highlight the second-best value.
}
    \begin{center}
        \newcolumntype{P}{>{\raggedleft\arraybackslash}p{.24in}}
        \resizebox{\textwidth}{!}{
\newrev{          \begin{tabular}{r P P P P P P P P P P P P P P P P P P}
 
            \toprule
 
             & \multicolumn{6}{c}{\textit{\buildtimeshort}   \textit{(ns/key)}}
             & \multicolumn{6}{c}{\textit{\queryshort} \textit{ (ns/key)}}
             & \multicolumn{6}{c}{\textit{\sizeshort} \textit{(\% of original size)}} \\
 
             \cmidrule(lr){2-7}
             \cmidrule(lr){8-13}
             \cmidrule(lr){14-19}
 
           \textit{Trie} &
           \textit{\wordsdata} & \textit{\urldata} & \textit{\dnadata} &
           \textit{\xmldata} & \textit{\wikidata} & \textit{\logdata} &
           \textit{\wordsdata} & \textit{\urldata} & \textit{\dnadata} &
           \textit{\xmldata} & \textit{\wikidata} & \textit{\logdata} &
           \textit{\wordsdata} & \textit{\urldata} & \textit{\dnadata} &
           \textit{\xmldata} & \textit{\wikidata} & \textit{\logdata} \\
 
            \midrule
 
FST
& 309&691&802&1157&716&2139
& 403&610&844&1541&1424&4003
& 45.5\%&35.5\%&74.3\%&39.2\%&37.9\%&39.5\%\\
 
\bvdesign-FST
& \cellcolor{blue!20}90&1043&1148&897&762&2173
& 228&404&562&835&1242&2197
& 40.4\%&39.8\%&75.9\%&39.7\%&38.4\%&40.4\%\\
 
\csquared-FST
& 101&1068&1114&894&729&1821
& 227&404&561&842&1258&2154
& 40.4\%&36.0\%&43.8\%&28.8\%&37.7\%&24.4\%\\
 
\midrule
 
 
CoCo'
& 790&1639&1875&1913&2114&7258
& 389&626&749&1027&1539&1597
& 51.5\%&41.7\%&73.2\%&40.6\%&43.7\%&40.2\%\\

\bvdesign-CoCo
& 758&1595&1919&1949&2145&7562
& 355&562&695&907&1320&1465
& 45.8\%&39.5\%&70.7\%&39.8\%&41.4\%&39.9\%\\
 
\csquared-CoCo
& 771&1637&1852&1900&2074&7237
& 358&560&613&898&1342&1448
& 45.8\%&35.7\%&38.5\%&28.9\%&40.7\%&23.9\%\\
 
\midrule
 
Marisa
& 234&470&676&712&521&1270
& 272&460&663&669&991&1098
& \cellcolor{blue!20}33.3\%&35.3\%&76.4\%&40.3\%&32.4\%&37.3\%\\
 
\bvdesign-Marisa
& 131&1046&1175&991&865&2031
& 178&\cellcolor{blue!20}310&443&\cellcolor{green!25}483&739&\cellcolor{blue!20}831
& 37.5\%&36.4\%&75.4\%&39.5\%&34.4\%&36.4\%\\
 
\csquared-Marisa
& 168&1064&1139&933&799&1613
& 182&\cellcolor{green!25}308&\cellcolor{blue!20}439&\cellcolor{blue!20}500&732&\cellcolor{green!25}829
& 37.5\%&30.9\%&41.0\%&27.6\%&32.9\%&20.7\%\\
 
Marisa-1
& 239&558&1002&945&599&1405
& 301&583&968&960&1244&1538
& \cellcolor{green!25}29.5\%&\cellcolor{green!25}22.0\%&34.3\%&25.2\%&\cellcolor{green!25}25.1\%&28.6\%\\
 
\bvdesign-Marisa-1
& 141&1229&2275&1476&1025&2209
& 178&425&686&746&967&1229
& 37.5\%&26.5\%&\cellcolor{blue!20}33.8\%&26.3\%&28.9\%&28.2\%\\
 
\csquared-Marisa-1
& 169&1294&2369&1473&997&1938
& 196&424&\cellcolor{green!25}436&764&979&1306
& 36.2\%&26.3\%&41.0\%&\cellcolor{blue!20}21.5\%&\cellcolor{blue!20}28.3\%&\cellcolor{blue!20}14.7\%\\
 
\midrule
 
PDT
& 624&1582&2947&4593&2040&3912
& 356&560&795&841&1022&1138
& 33.7\%&\cellcolor{blue!20}26.2\%&\cellcolor{green!25}28.4\%&\cellcolor{green!25}19.7\%&30.1\%&\cellcolor{green!25}13.0\%\\
 
ART
& \cellcolor{green!25}73&\cellcolor{green!25}83&\cellcolor{green!25}76&\cellcolor{green!25}115&\cellcolor{green!25}92&\cellcolor{green!25}289
& \cellcolor{green!25}132&347&533&728&\cellcolor{green!25}713&1611
& 346.8\%&156.2\%&115.0\%&59.1\%&165.7\%&27.6\%\\
 
C-ART
& 119&\cellcolor{blue!20}150&\cellcolor{blue!20}179&\cellcolor{blue!20}178&\cellcolor{blue!20}220&404
& \cellcolor{blue!20}148&401&674&750&810&1758
& 155.0\%&68.6\%&53.8\%&27.1\%&77.5\%&16.0\%\\
 
c-trie++
& 401&654&365&344&644&\cellcolor{blue!20}315
& 217&495&683&915&\cellcolor{blue!20}716&1507
& 1540.1\%&747.0\%&513.5\%&290.6\%&788.5\%&157.1\%\\
 
z-fast trie
& 447&779&989&1587&1168&2531
& 541&1158&1394&2016&1192&2838
& 1084.5\%&586.6\%&400.2\%&223.9\%&546.6\%&92.7\%\\

          \bottomrule
          \end{tabular}}
        }
    \end{center}
 
\label{tab:c2-full}
 
\end{table*}

%% file: latexfigs/table-trie-op.tex
\begin{table}[t]
  \caption{Bitvector Operation Latency (ns) on XML Dataset (2.1M keys).
  \textit{Speedup} = Baseline / \bvdesign. Bold entries indicate where \bvdesign
  outperforms the baseline.}
  \centering
  \small
 \newrev{ \begin{tabular}{l l r r r}
  \toprule
  \textit{Trie} & \textit{Operation} & \textit{Baseline (ns)} & \textit{\bvdesign (ns)} & \textit{Speedup}
  \\
  \midrule
  \multirow{4}{*}{FST}
    & \texttt{get}        & 1.48  & 1.52           & $0.98\times$ \\
    & \texttt{leaf\_id}   & 10.32 & \textbf{1.45}  & $\mathbf{7.13\times}$ \\
    & \texttt{degree}     & 1.09  & 2.75           & $0.40\times$ \\
    & \texttt{child\_pos} & 48.62 & \textbf{17.09} & $\mathbf{2.84\times}$ \\
  \midrule
  \multirow{5}{*}{CoCo}
    & \texttt{get}           & 0.64  & 1.06           & $0.60\times$ \\
    & \texttt{leaf\_id}      & 2.10  & \textbf{1.72}  & $\mathbf{1.23\times}$ \\
    & \texttt{internal\_id}  & 4.18  & \textbf{2.82}  & $\mathbf{1.48\times}$ \\
    & \texttt{degree}        & 1.83  & 2.34           & $0.78\times$ \\
    & \texttt{child\_pos}    & 22.60 & \textbf{17.87} & $\mathbf{1.26\times}$ \\
  \midrule
  \multirow{5}{*}{Marisa}
    & \texttt{get}         & 0.98  & 1.48           & $0.66\times$ \\
    & \texttt{link\_id}    & 8.20  & \textbf{1.23}  & $\mathbf{6.66\times}$ \\
    & \texttt{leaf\_id}    & 9.00  & \textbf{1.38}  & $\mathbf{6.53\times}$ \\
    & \texttt{child\_pos}  & 30.31 & \textbf{16.62} & $\mathbf{1.82\times}$ \\
    & \texttt{parent\_pos} & 28.68 & \textbf{16.29} & $\mathbf{1.76\times}$ \\
  \bottomrule
  \end{tabular}}
  \label{tab:trie-op}
\end{table}

%% file: conclusion.tex
\section{Conclusion}\label{sec:conclusion}

\balance This paper presents \csquared, guidelines for transforming
bitvector-based succinct tries to more cache-friendly and compact
versions. \csquared includes two optimizations: \bvdesign, the
\underline{C}ache-conscious bitvector design,
and 
\compressionscheme, the adaptive unary-path \underline{C}ompression.
Following \csquared, we redesign three state-of-the-art succinct tries:
\rev{FST, CoCo-trie, and Marisa}. \rev{On average, the three}
\csquared-optimized tries improve query time by
\newrev{\csqfstspeedup/\csqcocospeedup/\csqmarisaspeedup},
respectively, while using
\newrev{\csqfstspacesaving/\csqcocospacesaving/\csqmarisaspacesaving less space}.

Future work will 1) apply similar \bv optimizations to
DFUDS/BP tries, like PDT or the DFUDS variant of CoCo-trie, \mymarginpar{R2 W2,
  R2 D4}\rev{2) theoretically bound the space usage of practically-efficient
  compression schemes such as approximate \repair~\cite{pdt} and
  \fsst~\cite{fsst},} and 3) integrate other powerful text-compression
algorithms, such as LZW~\cite{lzw}, \rev{with} succinct tries.
